\documentclass[11pt,a4paper,reqno]{article}
\usepackage{a4wide}
\usepackage{microtype}
\usepackage{amsthm}
\usepackage{amsmath}
\usepackage{amssymb}
\usepackage{hyperref}
\hypersetup{colorlinks=true,citecolor=blue,linkcolor=blue,urlcolor=blue}
\usepackage[english]{babel}
\usepackage{enumerate}
\usepackage{footmisc}

\newcommand\Crestrict[2]{{% we make the whole thing an ordinary symbol
  \left.\kern-\nulldelimiterspace % automatically resize the bar with \right
  #1 % the function
  \right|_{#2} % this is the delimiter
  }}
  
\DeclareMathOperator{\reducesto}{\leq_{L}}
\DeclareMathOperator{\sa}{\leq_{SA}}

\DeclareMathOperator{\fpol}{fPol}
\DeclareMathOperator{\supp}{{\rm supp}}
\DeclareMathOperator{\bsupp}{{\bf supp}}
\DeclareMathOperator{\CSP}{CSP}
\DeclareMathOperator{\VCSP}{VCSP}
\DeclareMathOperator{\MinCSP}{Min-CSP}
\DeclareMathOperator{\MaxCSP}{Max-CSP}
\DeclareMathOperator{\Las}{Lasserre}
\DeclareMathOperator{\feas}{Feas}
\DeclareMathOperator{\opt}{Opt}
\DeclareMathOperator{\rdom}{{Feas}}
\DeclareMathOperator{\pol}{Pol}
\DeclareMathOperator*{\E}{\mathop{{\mathbb E}}}
\DeclareMathOperator{\ar}{ar}
\DeclareMathOperator{\avg}{{\mathrm avg}}

\newcommand{\eq}[1]{\ensuremath{\phi^{#1}_{=}}}
\renewcommand{\vec}[1]{\ensuremath{\mathbf{#1}}}

\newcommand{\inst}{\ensuremath{I}}
\newcommand{\tup}[1]{\ensuremath{\mathbf #1}}
\newcommand{\N}{\mbox{$\mathbb N$}}
\newcommand{\qq}{\ensuremath{\overline{\mathbb{Q}}}}
\newcommand{\lpopt}[2]{{\rm Opt_{LP}}}
\newcommand{\sdpopt}[2]{{\rm Opt_{SDP}}}
\newcommand{\sdpval}[0]{{\rm Val_{SDP}}}
\newcommand{\vcspopt}[0]{{\rm Opt_{VCSP}}}
\newcommand{\vcspval}[0]{{\rm Val_{VCSP}}}
\newcommand{\blambda}{\mbox{\boldmath $\lambda$}}
\newcommand{\bmu}{\mbox{\boldmath $\mu$}}

\newcommand{\bkappa}{\mbox{\boldmath $\kappa$}}

\date{}

\newtheorem{theorem}{Theorem}
\newtheorem{lemma}{Lemma} 
\newtheorem*{lemma*}{Lemma} 
\newtheorem*{proposition*}{Proposition} 
\newtheorem*{theorem*}{Theorem} 

\newtheorem{corollary}{Corollary}
\newtheorem{definition}{Definition}
\theoremstyle{remark}
\newtheorem{example}{Example}
\newtheorem{remark}{Remark}

\begin{document}
\title{The limits of SDP relaxations for general-valued CSPs}

\author{
Johan Thapper\\
Universit\'e Paris-Est, Marne-la-Vall\'ee, France\\
\texttt{thapper@u-pem.fr}
\and
Stanislav \v{Z}ivn\'{y}\thanks{An extended abstract of this work appeared in
\emph{Proceedings of the 32nd Annual ACM/IEEE Symposium on Logic in Computer
Science (LICS)}~\cite{tz17:lics}. This project has received funding from the
European Research Council (ERC) under the European Union's Horizon 2020 research
and innovation programme (grant agreement No 714532). The paper reflects only
the authors' views and not the views of the ERC or the European Commission. The
European Union is not liable for any use that may be made of the information
contained therein. Stanislav \v{Z}ivn\'y was supported by a Royal Society
University Research Fellowship.}\\
University of Oxford, UK\\
\texttt{standa.zivny@cs.ox.ac.uk}
}

\maketitle

\begin{abstract} 

It has been shown that for a general-valued constraint language $\Gamma$ the
following statements are equivalent: (1) any instance of $\VCSP(\Gamma)$ can be
solved to \emph{optimality} using a \emph{constant level} of the Sherali-Adams
LP hierarchy; (2) any instance of $\VCSP(\Gamma)$ can be solved to optimality
using the \emph{third level} of the Sherali-Adams LP hierarchy; (3) the support
of $\Gamma$ satisfies the ``\emph{bounded width condition}'', i.e., it contains
weak near-unanimity operations of all arities.

We show that if the support of $\Gamma$ violates the bounded width condition then
not only is $\VCSP(\Gamma)$ not solved by a constant level of the Sherali-Adams
LP hierarchy but it requires linear levels of the Lasserre SDP
hierarchy (also known as the sum-of-squares SDP hierarchy). For $\Gamma$
corresponding to linear equations in an Abelian group, this result follows from
existing work on inapproximability of Max-CSPs. By a breakthrough result of Lee,
Raghavendra, and Steurer [STOC'15], our result implies that for any $\Gamma$
whose support violates the bounded width condition no SDP relaxation of
polynomial-size solves $\VCSP(\Gamma)$.

We establish our result by proving that various reductions preserve exact solvability
by the Lasserre SDP hierarchy (up to a constant factor in the level of the
hierarchy). Our
results hold for general-valued constraint languages, i.e., sets of functions on
a fixed finite domain that take on rational or infinite values, and thus also
hold in notable special cases of $\{0,\infty\}$-valued languages (CSPs),
$\{0,1\}$-valued languages (Min-CSPs/Max-CSPs), and $\mathbb{Q}$-valued
languages (finite-valued CSPs). 

\end{abstract}

\section{Introduction}
\label{sec:intro}

\subsection{CSPs and exact solvability}

Constraint satisfaction problems (CSPs) constitute a broad class of
computational problems that involve assigning labels to variables subject to
constraints to be satisfied and/or optimised, as nicely explained in a survey by
Hell and Ne\v{s}et\v{r}il~\cite{Hell08:survey}. One line of
research focuses on CSPs parametrised by a set of (possibly weighted) relations
known as a constraint language~\cite{Jeavons97:jacm}. In their influential
paper, Feder and Vardi conjectured that for decision CSPs every constraint
language gives rise to a class of problems that belongs to P or is
NP-complete~\cite{Feder98:monotone}. The dichotomy conjecture of Feder and
Vardi has been verified in several
important special
cases by Schaefer~\cite{Schaefer78:complexity}, Hell and
Ne\v{s}et\v{r}il~\cite{Hell90:h-coloring},
Bulatov~\cite{Bulatov06:3-elementJACM,Bulatov11:conservative}, and Barto, Kozik,
and Niven~\cite{Barto09:siam}
mostly using the so-called algebraic
approach~\cite{Bulatov05:classifying,Barto14:jacm}.
Remarkably, the dichotomy conjecture has recently been solved independently by
Bulatov~\cite{Bulatov17:focs} and Zhuk~\cite{Zhuk17:focs}, respectively.

Using concepts from the extensions of the algebraic approach to optimisation
problems~\cite{cccjz13:sicomp}, the exact solvability of purely optimisation
CSPs, known as finite-valued CSPs, has been established by the authors~\cite{tz16:jacm} (these
include Min/Max-CSPs as a special case).
Putting together decision and optimisation problems in one framework, the exact
complexity of so-called general-valued CSPs has been
established, modulo the (now proved) classification of
decision CSPs, by the works of Kozik and Ochremiak~\cite{Kozik15:icalp} and
Kolmogorov, Krokhin, and Rol\'inek~\cite{Kolmogorov17:sicomp}. A result that proved useful when classifying both finite-valued
and general-valued CSPs is an algebraic characterisation of the power of the
basic linear programming relaxation for decision CSPs~\cite{kun12:itcs} and
general-valued CSPs~\cite{ktz15:sicomp}. 

\subsection{Approximation}

Convex relaxations, such as linear programming (LP) and semidefinite programming
(SDP), have long been powerful tools for designing efficient exact and approximation
algorithms~\cite{Vazirani2013approximation,Williamson2011design}. In particular,
for many combinatorial problems, the introduction of semidefinite programming relaxations allowed for
a new structural and computational
perspective~\cite{Goemans95:jacm,Karger98:jacm,Arora09:jacm}.
The Lasserre SDP hierarchy~\cite{Lasserre02:sjopt} is a sequence of semidefinite
relaxations for certain $0$-$1$ polynomial programs, each one more constrained
than the previous one. The $k$th level of the Lasserre SDP hierarchy requires any
set of $k$ variables of the relaxation, which live in a 
finite-dimensional real vector space, to be consistent in a
very strong sense. 
The $k$th level of the hierarchy can be solved in time
$L\cdot n^{O(k)}$, where $n$ is the number of variables and $L$ is the length of a binary encoding of the input. If an
integer program has $n$ variables then the $n$th level of the Lasserre SDP hierarchy
is tight, i.e., the only feasible solutions are convex combinations of integral
solutions. The Lasserre SDP hierarchy is similar in spirit to the
Lov\'asz-Schrijver SDP hierarchy~\cite{Lovasz91:jo} and the
Sherali-Adams LP hierarchy~\cite{Sherali1990}, but the Lasserre SDP hierarchy is
stronger~\cite{Laurent03:mor}. 

An important line of research, going back to a seminal work of
Yannakakis~\cite{Yannakakis91:jcss}, focuses on proving lower bounds on the size
of LP formulations.
Chan, Lee, Raghavendra, and Steurer~\cite{Chan16:jacm-approx} showed that Sherali-Adams LP relaxations are
universal for Max-CSPs in the sense that for every polynomial-size LP relaxation
of a Max-CSP instance $I$ there is a constant level of the Sherali-Adams LP
hierarchy of $I$ that achieves the same approximation guarantees. This result
has been improved to subexponential-size LP
relaxations by Kothari, Meka, and Raghavednra~\cite{Kothari17:stoc}. Moreover,
Ghosh and Tulsiani~\cite{Ghosh17:ccc} have shown that in
fact the basic LP relaxation enjoys the same universality property (among
super-constant levels of the Sherali-Adams LP hierarchy).
For related work on the integrality gaps for the Sherali-Adams LP
and Lov\'asz-Schrijver SDP hierarchies, we refer the reader
to~\cite{Schoenebeck08:focs,Charikar09:stoc,Chlamtac12:sdp}
and the references therein.

Recent years have seen some remarkable progress on lower bounds for the Lasserre
SDP hierarchy. 
Schoenebeck showed that certain problems require linear levels of the
Lasserre SDP
hierarchy~\cite{Schoenebeck08:focs}. In
particular, Schoenebeck showed, among other things, that $cn$ levels, for some
constant $0<c<1$,
of the Lasserre SDP hierarchy cannot prove that certain Max-CSPs (corresponding to
equations on the Boolean domain) are unsatisfiable~\cite{Schoenebeck08:focs}. Tulsiani extended this work 
to Max-CSPs corresponding to equations over Abelian groups of prime
orders~\cite{Tulsiani09:stoc}. Finally, Chan extended this to Max-CSPs corresponding to equations
over Abelian groups of arbitrary size~\cite{Chan16:jacm}.
In a recent breakthrough, Lee, Raghavendra, and Steurer~\cite{Lee15:stoc} showed that the Lasserre 
SDP relaxations are universal for Max-CSPs in the sense that for every
polynomial-size SDP relaxation of a Max-CSP instance $I$ there is a constant
level of the Lasserre SDP hierarchy of $I$ that achieves the same approximation
guarantees. One of the many ingredients of the proof in~\cite{Lee15:stoc} is to view the Lasserre
SDP hierarchy as the Sum-of-Squares
algorithm~\cite{Lasserre01:jo},
which relates to proof complexity~\cite{ODonnell13:soda-approximability}. (In fact, Schoenebeck's
above-mentioned result had independently been obtained by
Grigoriev~\cite{Grigoriev01:tcs} using this view.)

\subsection{Bounded width condition}

We now informally describe the bounded width condition (BWC). A set of
operations on a fixed finite domain satisfies the BWC if it contains ``weak
near-unanimity'' operations of all possible arities. An operation is called a
weak near-unanimity operation if it is symmetric when all the arguments but one
are the same. (A formal definition is given in Section~\ref{sec:sa}.) An example
of a ternary weak-near unanimity operation is a majority operation, which
satisfies $f(x,x,y)=f(x,y,x)=f(y,x,x)=x$ for all $x$ and $y$.
Polymorphisms~\cite{Bulatov05:classifying}, which are at the heart of the
algebraic approach to CSPs, are operations that combine satisfying assignments
to a CSP instance and produce a new satisfying assignment. We say that a CSP
instance $I$ satisfies the BWC if the set of all polymorphisms of $I$ satisfies
the BWC.

In an important series of papers by Mar\'oti and
McKenzie~\cite{Maroti08:weakly}, Larose and Z\'adori~\cite{LaroseZadori07:au},
Barto and Kozik~\cite{Barto14:jacm}, and Bulatov~\cite{Bulatov16:lics}, it was
established that the BWC captures precisely the decision CSPs that are solved by
Datalog, a natural and well-studied local propagation
algorithm~\cite{Feder98:monotone}.

\subsection{Contributions}

In our previous work~\cite{tz17:sicomp} (which we refer the reader to for more
information and background), we studied the power of the Sherali-Adams LP
hierarchy for exact solvability of general-valued CSPs. In particular, we have
shown in~\cite{tz17:sicomp} that general-valued CSPs that are solved exactly by
a constant level of the Sherali-Adams LP hierarchy are precisely those
general-valued CSPs that satisfy the BWC.
In more detail, fractional polymorphisms of a general-valued CSP instance $I$
are probability distributions over polymorphisms of $I$ that in a sense preserve
the weighted relations of $I$.
For a constraint language $\Gamma$, we denote by $\supp(\Gamma)$ the
set of operations that appear in the support of some fractional polymorphism of
$\Gamma$. (Formal definitions are given in Section~\ref{sec:prelims}.)  The
following theorem is the main result of~\cite{tz17:sicomp}.

\begin{theorem}[\protect{\cite[Theorem~3.3]{tz17:sicomp}}]\label{thm:sa}
Let $\Gamma$ be a general-valued constraint language of finite size. The following are equivalent:
\begin{enumerate}[(i)]
\item $\VCSP(\Gamma)$ is solved by a \emph{constant} level of the Sherali-Adams LP hierarchy. \label{cnd:bound}
\item $\VCSP(\Gamma)$ is solved by the \emph{third} level of the Sherali Adams LP hierarchy. \label{cnd:23}
\item $\supp(\Gamma)$ satisfies the BWC.  \label{cnd:BWC}
\end{enumerate}
\end{theorem}

In this follow-up work, we study the power of the Lasserre SDP hierarchy for
exact solvability of general-valued CSPs. As our main contribution (stated as
Theorem~\ref{thm:main}), we show that general-valued CSPs that are not solved by
a constant level of the Sherali-Adams LP hierarchy require linear levels of the 
Lasserre SDP hierarchy. As a direct corollary, the
results of Lee, Raghavendra, and Steurer~\cite{Lee15:stoc} imply that such
general-valued CSPs are not solved by \emph{any} polynomial-size SDP relaxation. 

In order to prove our result, we will strengthen the proof of the
implication $(i)\Longrightarrow (iii)$ of Theorem~\ref{thm:sa}.
The idea is to show that if
$\supp(\Gamma)$ violates the BWC, then $\Gamma$ can \emph{simulate}
linear equations in some Abelian group. It suffices to show that linear
equations require linear levels of the Lasserre SDP hierarchy and that
the simulation preserves exact solvability by the Lasserre SDP hierarchy (up to a
constant factor in the level of the hierarchy).
As discussed before, the former is actually known 
(in a stronger sense of inapproximability of linear
equations)~\cite{Grigoriev01:tcs,Schoenebeck08:focs,Tulsiani09:stoc,Chan16:jacm}
and will be discussed in Section~\ref{subsec:overview}. Our contribution is proving
the latter. 
While the simulation involves only local replacements via gadgets, it needs to
be done with care. In particular, we emphasise that the simulation involves
steps, such as going to the core and interpretations, which are commonly used in
the algebraic approach to CSPs but not in the literature on convex relaxations
and approximability of CSPs~\cite{Tulsiani09:stoc}. Indeed, the algebraic
approach to CSPs gives the right tools for the intuitive (but non-trivial
to capture formally) meaning of ``simulating equations''.

\subsection{Related work}

In our main result, Theorem~\ref{thm:main}, the BWC is required to hold, as in
Theorem~\ref{thm:sa}, for the support of the fractional
polymorphisms~\cite{cccjz13:sicomp} of the general-valued CSPs. This is a
natural requirement since polymorphisms do not capture the complexity of
general-valued CSPs but the fractional polymorphisms do
so~\cite{cccjz13:sicomp,Kolmogorov17:sicomp}.

The BWC was also shown~\cite{Dalmau13:robust,Barto16:sicomp} to capture
precisely the Max-CSPs that can be robustly approximated, as
conjectured by Guruswami and Zhou~\cite{Guruswami12:toc}. This work is similar to ours but
different. In particular, Dalmau and Krokhin showed~\cite{Dalmau13:robust} that various 
reductions preserve robust approximability of equations, and thus showing
that Max-CSPs not satisfying the BWC cannot be robustly approximated, assuming
P$\neq$NP and relying on H{\aa}stad's inapproximability results for linear
equations~\cite{Hastad01:jacm}. (Barto and
Kozik~\cite{Barto16:sicomp} then showed that Max-CSPs satisfying the BWC can be
robustly approximated.)
However, note that linear equations \emph{can} be solved exactly using Gaussian
elimination and thus this result is not applicable in our setting.
Our result, on the other hand, shows that various reductions preserve
exact solvability of equations by a \emph{particular} algorithm (the Lasserre
SDP hierarchy) independently of {P\,vs.\,NP}. Moreover, the pp-definitions and
pp-interpretations used in~\cite{Dalmau13:robust,Barto16:sicomp} were required
to be equality-free. We prove that our reductions are well-behaved without this
assumption. 

Our main result is incomparable with the results obtained by
Schoenebeck~\cite{Schoenebeck08:focs}, Tulsiani~\cite{Tulsiani09:stoc}, and
Chan~\cite{Chan16:jacm} in the context of (in)approximability.
On the one hand, our results capture exact solvability
rather than approximability. On the other hand, we give a stronger result as our
result applies to general-valued CSPs rather than only to Max-CSPs or
finite-valued CSPs.
General-valued CSPs are more expressive than their special cases Max-CSPs 
and finite-valued CSPs since general-valued CSPs also include decision CSPs as a special case and thus can
use ``hard'' or ``strict'' constraints. 
The results on Max-CSPs~\cite{Schoenebeck08:focs,Tulsiani09:stoc,Chan16:jacm}
were extended by (problem-specific) reductions to some problems (such as Vertex
Cover) which are not captured by Max-CSPs but are captured by general-valued
CSPs. 
Our results are not problem specific and apply to \emph{all} general-valued CSPs. In
particular, we give a \emph{complete} characterisation of which
\emph{general-valued} CSPs are solved exactly by the Lasserre SDP hierarchy.

Our results generalise some of the results of Dawar and Wang~\cite{Dawar17:lics} and
Atserias and Ochremiak~\cite{Atserias17:icalp}. In particular, using
definability in counting logics, Dawar and Wang have established our main result
in the special case of $\mathbb{Q}$-valued languages, i.e., for finite-valued
CSPs~\cite{Dawar17:lics}. Moreover, using tools from proof complexity, Atserias
and Ochremiak have established (among other things) our main result in the
special case of $\{0,\infty\}$-valued languages, i.e., for (decision)
CSPs~\cite{Atserias17:icalp}.

\section{Preliminaries}
\label{sec:prelims}

\subsection{General-valued CSPs}
We first describe the framework of general-valued constraint satisfaction problems (VCSPs). 
Let $\qq=\mathbb{Q}\cup\{\infty\}$ denote the set of rational numbers extended
with positive infinity. Throughout the paper, let $D$ be a fixed finite set of
size at least two, also called a \emph{domain}; we call the elements of $D$
\emph{labels}.
We denote by $\left[n\right]$ the set $\{1,\ldots,n\}$.

\begin{definition}\label{def:wrel}
An $r$-ary \emph{weighted relation} over $D$ is a mapping $\phi:D^r\to\qq$. We
write $\ar(\phi)=r$ for the arity of $\phi$.
\end{definition}

A weighted relation $\phi \colon D^r \to \{0,\infty\}$ can be seen as the
(ordinary) relation $\{ \tup{x} \in D^r \mid \phi(\tup{x}) = 0 \}$. We will use
both viewpoints interchangeably.

For any $r$-ary weighted relation $\phi$, we denote by $\feas(\phi)=\{\tup{x}\in
D^r \mid \phi(\tup{x})<\infty\}$ the underlying $r$-ary \emph{feasibility
relation}, and by $\opt(\phi)=\{\tup{x}\in\feas(\phi) \mid \forall\tup{y}\in
D^r: \phi(\tup{x})\leq\phi(\tup{y})\}$ the $r$-ary \emph{optimality relation},
which contains the tuples on which $\phi$ is minimised.

\begin{definition}
Let $V=\{x_1,\ldots, x_n\}$ be a set of variables. A \emph{valued constraint} over $V$ is an expression
of the form $\phi(\tup{x})$ where $\phi$ is a weighted relation and $\tup{x}\in V^{\ar(\phi)}$.
The tuple $\tup{x}$ is called the \emph{scope} of the constraint.
\end{definition}

\begin{definition}\label{def:vcsp}
An instance $I$ of the \emph{valued constraint satisfaction problem} (VCSP) is specified
by a finite set $V=\{x_1,\ldots,x_n\}$ of variables, a finite set $D$ of labels,
and an \emph{objective function} $\phi_\inst$
expressed as follows:
\begin{equation*}
\phi_\inst(x_1,\ldots, x_n)=\sum_{i=1}^q{\phi_i(\tup{x}_i)},
\end{equation*}
where each $\phi_i(\tup{x}_i)$, $1\le i\le q$, is a valued constraint. 
Each constraint may appear multiple times in~$\inst$.
An \emph{assignment} to $\inst$ is a map $\sigma \colon V \to D$.
The goal is to find an assignment that minimises the objective function.
\end{definition}

For a VCSP instance $I$, we write $\vcspval(\inst,\sigma)$ for
$\phi_\inst(\sigma(x_1), \dots, \sigma(x_n))$, and $\vcspopt(\inst)$ for the
minimum of $\vcspval(\inst,\sigma)$ over all assignments $\sigma$.

An assignment $\sigma$ with $\vcspval(\inst,\sigma)<\infty$ is called
\emph{satisfying}.
An assignment $\sigma$ with $\vcspval(\inst,\sigma)=\vcspopt(\inst)$ is called \emph{optimal}. 

A VCSP instance $I$ is called \emph{satisfiable} if there is a satisfying
assignment to $I$. Constraint satisfaction problems (CSPs) are a special case of VCSPs with (unweighted) relations
with the goal to determine the existence of a satisfying assignment.

A \emph{general-valued constraint language} 
(or just a \emph{constraint language} for short) 
over $D$ is a set of weighted relations over $D$. As is common in the (V)CSP
literature, we will focus on constraint languages of \emph{finite} size.
We denote by $\VCSP(\Gamma)$
the class of all VCSP instances in which the weighted relations are all
contained in $\Gamma$. A constraint language $\Gamma$ is called \emph{crisp} if
$\Gamma$ contains only (unweighted) relations. For a crisp language $\Gamma$,
$\VCSP(\Gamma)$ is equivalent to the well-studied (decision)
$\CSP(\Gamma)$~\cite{Hell08:survey}.
We remark that for $\{0,1\}$-valued
constraint languages, $\VCSP(\Gamma)$ is also known as $\MinCSP(\Gamma)$ or
$\MaxCSP(\Gamma)$ (since for exact solvability these are equivalent).

For a constraint language $\Gamma$, let $\ar(\Gamma)$ denote $\max \{
\ar(\phi) \mid \phi \in \Gamma \}$.

\begin{example}\label{ex:langs}
Let $D=\{0,1\}$. We define several weighted relations. 
\begin{itemize}
\item $\phi_{\sf cut}(x,y)=1$ if $x+y=0 \pmod{2}$ and $\phi_{\sf cut}(x,y)=0$ otherwise.
\item $\phi_{\sf mc}(x,y)=1$ if $x+y=1 \pmod{2}$ and $\phi_{\sf mc}(x,y)=0$ otherwise.
\item For $a\in D$, $c_a(x)=0$ if $x=a$ and $c_a(x)=\infty$ otherwise.
\item For $a\in D$, $R_a(x,y,z)=0$ if $x+y+z=a \pmod{2}$ and $R_a(x,y,z)=\infty$ otherwise.
\end{itemize}
Let $\Gamma_{\sf cut}=\{\phi_{\sf cut}, c_0, c_1\}$, 
$\Gamma_{\sf mc}=\{\phi_{\sf mc}\}$, and $\Gamma_{\sf eq}=\{R_0, R_1\}$.
Then, $\VCSP(\Gamma_{\sf cut})$ corresponds to the $(s,t)$-Min-Cut problem, 
$\VCSP(\Gamma_{\sf mc})$ corresponds to the Min-UnCut problem, 
and finally $\VCSP(\Gamma_{\sf eq})$ corresponds to the feasibility problem for 
systems of linear questions in three variables over ${\mathbb Z}_2$.
\end{example}

\subsection{Fractional polymorphisms} 
We next define fractional polymorphisms, which are algebraic properties known to
capture the computational complexity of the underlying class of VCSPs. 

Given an $r$-tuple $\tup{x}\in D^r$, we denote its $i$th entry by $\tup{x}[i]$ for $1\leq i\leq r$.
A mapping $f \colon D^m\rightarrow D$ is called an $m$-ary \emph{operation} on $D$; $f$ is
\emph{idempotent} if $f(x,\ldots,x)=x$.
We apply an $m$-ary operation $f$ to $m$ $r$-tuples
$\tup{x_1},\ldots,\tup{x_m}\in D^r$ coordinatewise, that is, 
$f(\tup{x_1},\ldots,\tup{x_m})=(f(\tup{x_1}[1],\ldots,\tup{x_m}[1]),\ldots,f(\tup{x_1}[r],\ldots,\tup{x_m}[r]))$.

\begin{definition} \label{def:pol}
Let $\phi$ be a weighted relation on $D$ and let $f$ be an $m$-ary operation on $D$.
We call $f$ a \emph{polymorphism of $\phi$} if,
for any $\tup{x_1},\ldots,\tup{x_m} \in \rdom(\phi)$,
we have that $f(\tup{x_1},\ldots,\tup{x_m})\in\rdom(\phi)$.

For a constraint language $\Gamma$,
we denote by $\pol(\Gamma)$ the set of all operations which are polymorphisms of all 
$\phi \in \Gamma$. We write $\pol(\phi)$ for $\pol(\{\phi\})$.
\end{definition}

The intuition behind polymorphisms is that if $\pol(\Gamma)$ contains only
``trivial'' operations (such as projections, cf. Example~\ref{ex:fpols}) then
checking for a satisfiable solution to an instance of $\VCSP(\Gamma)$ is NP-hard,
whereas if $\pol(\Gamma)$ contains a ``non-trivial'' operation then this can be
done in polynomial time. This intuition was formalised in the algebraic
dichotomy conjecture~\cite{Bulatov05:classifying} recently proved
in~\cite{Bulatov17:focs,Zhuk17:focs}.

The following notions are known to capture the complexity of general-valued constraint
languages~\cite{cccjz13:sicomp,Kozik15:icalp} and will also be important in this
paper.
A probability distribution $\omega$ over the set of $m$-ary operations on $D$ is
called an $m$-ary \emph{fractional operation}. For a
fractional operation $\omega$, ``$f\sim\omega$''
means that $f$ is a random operation (of the same arity as $\omega$) drawn
according to the distribution $\omega$. We define $\supp(\omega)$ to
be the set of operations assigned positive probability by $\omega$. We denote by
$\avg$ the average operator; i.e.,
$\avg\{a_1,\ldots,a_m\}=(1/m)\sum_{i=1}^ma_i$.

\begin{definition} \label{def:wp} 
Let $\phi$ be a weighted relation on $D$ and
let $\omega$ be an $m$-ary fractional operation on $D$.
We call $\omega$ a \emph{fractional polymorphism of $\phi$} if 
$\supp(\omega)\subseteq\pol(\phi)$ and for any
$\tup{x}_1,\ldots,\tup{x}_m \in \rdom(\phi)$, we have
\begin{equation*}
\E_{f\sim \omega}[\phi(f(\vec{x}_1,\ldots,\vec{x}_m))]\ \le\
\avg\{\phi(\vec{x}_1),\ldots,\phi(\vec{x}_m)\}.
\label{eq:wpol}
\end{equation*}
For a general-valued constraint language $\Gamma$, we denote by $\fpol(\Gamma)$ the set of all
fractional operations which are fractional polymorphisms of all weighted
relations $\phi \in \Gamma$. We write $\fpol(\phi)$ for $\fpol(\{\phi\})$.
\end{definition}

In case of fractional polymorphisms, the important operations are those that are
assigned positive probability.

\begin{definition}
Let $\Gamma$ be a general-valued constraint language on $D$. We define
\begin{equation*}
\supp(\Gamma)\ =\ \bigcup_{\omega \in \fpol(\Gamma)} \supp(\omega).
\end{equation*}
\end{definition}

The intuition behind fractional polymorphisms is that if $\supp(\Gamma)$
contains only ``trivial'' operations then finding an optimal solution to an
instance of $\VCSP(\Gamma)$ is NP-hard, whereas if $\supp(\Gamma)$ contains a
``non-trivial'' operation then this can be done in polynomial time. This
intuition was formalised in~\cite{cccjz13:sicomp,Kozik15:icalp} and proved
in~\cite{Kolmogorov17:sicomp}. We now give some examples.

\begin{example}\label{ex:fpols}
Let $D=\{0,1\}$ and recall the constraint languages $\Gamma_{\sf cut}$, $\Gamma_{\sf
mc}$, and $\Gamma_{\sf eq}$ defined in Example~\ref{ex:langs}. 

Consider the two binary operations $\min$ and $\max$ on $D$ that return the
smaller and the larger of its two arguments, respectively. The constraint
language $\Gamma_{\sf cut}$ admits $\omega_{\sf sub}$ as a fractional
polymorphism, where $\omega_{\sf sub}(\min)=\omega_{\sf
sub}(\max)=\frac{1}{2}$. In fact, the set of all weighted relations that
admit $\omega_{\sf sub}$ as a fractional polymorphism is precisely the
class of \emph{submodular} functions.
Note that both $\min$ and $\max$ are binary commutative operations.
By~\cite[Corollary~6]{ktz15:sicomp}, the fact that $\supp(\Gamma_{\sf cut})$
contains a binary commutative operation implies that 
$\VCSP(\Gamma_{\sf cut})$ is solved by the first level of the Sherali-Adams LP hierarchy.

Since $\VCSP(\Gamma_{\sf mc})$ is essentially the problem Min-UnCut,
it is NP-hard.
This fact can also be deduced from looking at the binary fractional polymorphisms of
$\Gamma_{\sf mc}$.
For $i\in\{1,2\}$, we denote by $\pi_i$ the binary operation that returns its
$i$th argument (these are known as projections). Also, for
$i\in\{1,2\}$, we denote by $\pi'_i$ the binary operation defined by
$\pi'_i(0,0)=1$, $\pi'_i(1,1)=0$, and $\pi'_i(x,y)=\pi_i(x,y)$ for
$x\neq y$. 
For any $0\leq p\leq \frac{1}{2}$, the binary fractional operation
$\omega_p$ defined by $\omega_p(\pi_1)=\omega_p(\pi_2)=p$ and
$\omega_p(\pi'_1)=\omega_p(\pi'_2)=\frac{1}{2}-p$ is a fractional
polymorphism of $\Gamma_{\sf mc}$.
It is not hard to show that all fractional polymorphisms of $\Gamma_{\sf mc}$
are of this form, and hence there is no binary commutative operation in $\supp(\Gamma_{\sf mc})$.
It then follows from~\cite{ktz15:sicomp} that $\VCSP(\Gamma_{\sf mc})$ is not
solved by the first level of the Sherali-Adams LP hierarchy, and by the
results in~\cite{tz16:jacm} that $\VCSP(\Gamma_{\sf mc})$ is NP-hard.

Finally, let $m$ denote the ternary operation defined by $m(x,y,z)=x+y+z
\pmod{2}$. The constraint language $\Gamma_{\sf eq}$ admits $m$ as a
polymorphism and thus any instance of $\VCSP(\Gamma_{\sf eq})$ can be
solved in polynomial time~\cite{Jeavons97:jacm}. 
However, $\pol(\Gamma_{\sf eq})$ does not contain any weak near-unanimity operation
of arity 3 (defined in Section~\ref{sec:sa}).
It therefore follows from Thereom~\ref{thm:main} of this paper that
$\VCSP(\Gamma_{\sf eq})$
requires linear levels of the Lasserre SDP hierarchy.
\end{example}

\subsection{Expressibility, interpretability, and simulation}

In this section we formally define the various types of gadget constructions
needed to establish our main result. We also introduce the important notion of cores.

\begin{definition} \label{def:expres}
We say that an $m$-ary weighted relation $\phi$ is \emph{expressible} over a
general-valued constraint
language $\Gamma$ if there exists an instance $\inst$ of $\VCSP(\Gamma)$ 
with variables $x_1,\ldots,x_m,v_1,\ldots,v_p$
such that
\begin{equation*}
\phi(x_1,\ldots,x_m) = \min_{v_1,\dots,v_p} \phi_I(x_1,\dots,x_m,v_1,\dots,v_p).
\end{equation*}
\end{definition}

For a fixed set $D$, let $\eq{D}$ denote the binary equality relation $\{ (x,x)
\mid x \in D \}$.
We denote by $\langle \Gamma \rangle$ the set of weighted relations obtained by
taking the closure of $\Gamma \cup \{\eq{D}\}$, where $D$ is the domain of
$\Gamma$, under expressibility, the $\feas$ and $\opt$ operations, scaling by
nonnegative rational constants, and addition of rational constants.

\begin{definition}\label{def:inter}
Let $\Gamma$ and $\Delta$ be general-valued constraint languages on domain $D$ and $D'$, respectively.
We say that
$\Delta$ \emph{ has an interpretation in} $\Gamma$ with parameters $(d,S,h)$ if there exists a
$d \in \N$, a set $S\subseteq D^d$, and a surjective map $h:S\to D'$ such
that $\langle \Gamma \rangle$ contains the following weighted relations:
\begin{itemize}
\item $\phi_S \colon D^d \to \qq$ defined by
$\phi_S(\tup{x}) = 0$ if $\tup{x} \in S$ and $\phi_S(\tup{x}) = \infty$ otherwise;
\item $h^{-1}(\eq{D'})$; and
\item $h^{-1}(\phi_i)$, for every weighted relation $\phi_i \in \Delta$, 
\end{itemize}
where $h^{-1}(\phi_i)$, for an $m$-ary weighted relation $\phi_i$,
is the $dm$-ary weighted relation on $D$ defined by
$h^{-1}(\phi_i)(\tup{x}_1,\dots,\tup{x}_m) = \phi_i(h(\tup{x}_1),\dots,h(\tup{x}_m))$,
for all $\tup{x}_1,\dots,\tup{x}_m\in~S$.
\end{definition}

It follows from Definition~\ref{def:inter} that interpretations compose.

\begin{remark}
A weighted relation being expressible over $\Gamma \cup \{ \eq{D} \}$ is the
analogue of a relation being definable by a \emph{primitive positive (pp)}
formula (using existential quantification and conjunction) over a relational
structure with equality. Indeed, when $\Gamma$ is crisp, the two notions
coincide. Also, for a crisp $\Gamma$ the notion of an interpretation coincides
with the notion of a \emph{pp-interpretation} for relational
structures~\cite{Bodirsky08:survey}.
\end{remark}

For a subset of the domain $S\subseteq D$, we define the restriction of a language $\Gamma$ on $S$ as follows.

\begin{definition}
Let $\Gamma$ be a general-valued constraint language with domain $D$ and let $S \subseteq D$.
The \emph{sub-language $\Gamma[S]$ of $\Gamma$ induced by $S$} is the constraint
language defined on domain $S$ and containing the restriction of every weighted relation
$\phi\in\Gamma$ onto $S$.
\end{definition}

Appropriate notions of cores have played an important role in the complexity classification of CSPs~\cite{Bulatov05:classifying,Bulatov17:focs,Zhuk17:focs} and VCSPs~\cite{Kozik15:icalp,Kolmogorov17:sicomp}.
We define a core based on the unary operations in the support of a language, as is done in~\cite{tz17:sicomp,Kolmogorov17:sicomp}.

\begin{definition}\label{def:core}
A general-valued constraint language $\Gamma$ is \emph{a core} if all unary operations in
$\supp(\Gamma)$ are bijections.
A general-valued constraint language $\Gamma'$ is a \emph{core of $\Gamma$} if $\Gamma'$ is a
core and $\Gamma' = \Gamma[f(D)]$ for some unary $f \in \supp(\Gamma)$.
\end{definition}

We can now give a formal definition of the notion of \emph{simulation} used
in the statement of our main result, Theorem~\ref{thm:main}.
Recall from Example~\ref{ex:langs} that $c_a$ denotes the constant unary
relation containing the label $a$.
Let $\mathcal{C}_D = \{ c_a \mid a \in D \}$ be the set of all constant unary relations on the set $D$. 

\begin{definition}
Let $\Gamma'$ be a core of a general-valued constraint language $\Gamma$ on domain $D' \subseteq D$.
We say that \emph{$\Gamma$ can simulate} a general-valued constraint language $\Delta$ if 
$\Delta$ has an interpretation in $\Gamma' \cup \mathcal{C}_{D'}$.
\end{definition}

We note that simulation is known to preserve polynomial-time
solvability~\cite{Bulatov05:classifying,cccjz13:sicomp,tz16:jacm,Kozik15:icalp}.
We will show later, in Theorem~\ref{thm:reductions}, that simulation
additionally preserves exact solvability in the Lasserre SDP hierarchy, defined
in Section~\ref{sec:lass}, up to a constant factor in the level of the
hierarchy.

\section{Lower Bounds on LP and SDP Relaxations}

Every VCSP instance has a natural LP relaxation known as the \emph{basic LP
relaxation} (BLP). The power of BLP for exact solvability of $\CSP(\Gamma)$,
where $\Gamma$ is a crisp constraint language, has been characterised (in terms
of the polymorphisms of $\Gamma$) in~\cite{kun12:itcs}. The power of BLP for
exact solvability of $\VCSP(\Gamma)$, where $\Gamma$ is a general-valued
constraint language, has been characterised (in terms of the fractional
polymorphisms of $\Gamma$) in~\cite{ktz15:sicomp}.

The Sherali-Adams LP hierarchy~\cite{Sherali1990} gives a systematic way of
strengthening the BLP relaxation. BLP being the first level, the $k$th level of
the Sherali-Adams LP hierarchy adds to the BLP linear constraints satisfied by
the integral solutions and involving at most $k$ variables. One can think of the
variables of the $k$th level as probability distributions over assignments to at
most $k$ variables of the original instance.

The Lasserre SDP hierarchy~\cite{Lasserre02:sjopt} is a significant
strengthening of the Sherali-Adams LP hierarchy: real-valued variables are
replaced by vectors from a finite-dimensional real vector space. Intuitively,
the norms of these vectors again induce probability distributions over
assignments to at most $k$ variables of the original instance (for the $k$th
level of the Lasserre SDP hierarchy). Since these distributions have to come
from inner products of vectors, this is a tighter relaxation. In particular, it
is known that the $k$th level of the Lasserre SDP hierarchy is at least as tight
as the $k$th level of the Sherali-Adams LP hierarchy~\cite{Laurent03:mor}. 

It is well known that for a problem with $n$ variables, the $n$th levels of both
of these two hierarchies are exact, i.e., the solutions to the $n$th levels are
precisely the convex combinations of the integral solutions. However, it is not
clear how to solve the $n$th levels in polynomial time. In general, taking an
$n$-variable instance of $\VCSP(\Gamma)$, the $k$th
level of both hierarchies can be solved in time 
$L\cdot n^{O(k)}$, where $L$ is the length of a binary encoding of the input.
In particular, this is polynomial for a fixed $k$.

In this section, we will define the Sherali-Adams LP and the Lasserre SDP
hierarchies and state known and new results regarding their power and
limitations for exact solvability of general-valued CSPs.

\subsection{Sherali-Adams LP Hierarchy}
\label{sec:sa}

Let $I$ be an instance of the VCSP with $\phi_\inst(x_1,\dots,x_n) =
\sum_{i=1}^q \phi_i(\tup{x}_i)$, $X_i \subseteq V = \{x_1, \dots, x_n\}$ and
$\phi_i \colon D^{\ar(\phi_i)} \to \qq$. We will use the notational convention
to denote by $X_i$ the \emph{set} of variables occurring in the scope
$\tup{x}_i$.

A \emph{null constraint} on a set $X \subseteq V$ is a constraint with a weighted
relation identical to $0$.
It is sometimes convenient to add null constraints to a VCSP instance as
placeholders, to ensure that they have scopes where required, even if these relations may not
necessarily be members of the corresponding constraint language $\Gamma$. In
order to obtain an equivalent instance that is formally in $\VCSP(\Gamma)$, the
null constraints can simply be dropped, as they are always satisfied and do not
influence the value of the objective function.

Let $k$ be an integer. The $k$th level of the Sherali-Adams LP
hierarchy~\cite{Sherali1990}, henceforth called the SA$(k)$-relaxation of $I$, is given by
the following linear program.
Ensure that for every non-empty $X
\subseteq V$ with $|X| \leq k$ there is some constraint $\phi_i(\tup{x}_i)$
with $X_i=X$, possibly by adding null constraints.
The variables of the SA$(k)$-relaxation, given in Figure~\ref{fig:lp}, are $\lambda_{i}(\sigma)$ for every $i \in \left[q\right]$ and assignment $\sigma \colon X_i \to D$.
We slightly abuse notation by writing $\sigma \in \feas(\phi_i)$ for $\sigma
\colon X_i \to D$ such that $\sigma(\tup{x}_i) \in \feas(\phi_i)$. 
\begin{figure*}[t]
\begin{alignat}{2}
\text{minimise}\ &  \sum_{i = 1}^q \sum_{\sigma \in \feas(\phi_i)}  \lambda_{i}(\sigma) \phi_i(\sigma(\tup{x}_i)) \nonumber\\
\text{subject to}\ & & \nonumber\\
& \lambda_{i}(\sigma) \geq 0 &\ &\ \forall i \in \left[q\right], \sigma \colon
X_i \to D \label{sa:nonnegative} \tag{S1}\\
& \lambda_{i}(\sigma) = 0 &\ &\ \forall i \in \left[q\right], \sigma \colon X_i
\to D, \sigma(\tup{x}_i) \not\in \feas(\phi_i) \label{sa:infeas}\tag{S2} \\
& \sum_{\sigma \colon X_i \to D} \lambda_{i}(\sigma) = 1 &\ &\ \forall
i\in\left[q\right] \label{sa:sum1}\tag{S3} \\
&  \sum_{\substack{\sigma \colon X_i \to D\\ \Crestrict{\sigma}{X_j} = \tau}} \lambda_i(\sigma) = \lambda_{j}(\tau) &\ &\ \forall i, j \in \left[q\right] : X_j \subseteq X_i, \left|X_j\right| \leq k, \tau \colon X_j \to D \nonumber \label{sa:marginal}\tag{S4}
\end{alignat}
\vspace*{-1em}
\caption{The $k$th level of the Sherali-Adams LP hierarchy, SA$(k)$.}\label{fig:lp}
\end{figure*}

We write $\lpopt{}{}(\inst,k)$ for the optimal value of an LP-solution to the
SA$(k)$-relaxation of $\inst$.

\begin{definition}
Let $\Gamma$ be a general-valued constraint language.
We say that $\VCSP(\Gamma)$ is \emph{solved by the
$k$th level of the Sherali-Adams LP hierarchy}
if for every instance $I$ of $\VCSP(\Gamma)$ we have
$\vcspopt(\inst)=\lpopt{}{}(I,k)$.
\end{definition}

We now describe the main result from~\cite{tz17:sicomp}, which captures the
power of Sherali-Adams LP relaxations for exact optimisation of VCSPs.

An $m$-ary idempotent operation $f \colon D^m\to D$ is called a 
\emph{weak near-unanimity} (WNU) operation if, for all $x,y\in D$,
\begin{equation}\label{pseudownu}
f(y,x,x,\ldots,x)=f(x,y,x,x,\ldots,x)=\cdots =f(x,x,\ldots,x,y).\tag{WNU}
\end{equation}

\begin{definition}\label{def:bwc} 
A set of operations satisfies the \emph{bounded width condition (BWC)}
if it contains a (not necessarily idempotent) $m$-ary operation satisfying the identities~(\ref{pseudownu}), for every $m\geq 3$.
\end{definition}

Recall from Section~\ref{sec:intro} Theorem~\ref{thm:sa}, which characterises the power of constant levels of
the Sherali-Adams LP hierarchy for exact solvability of VCSPs in terms of the
BWC.

\begin{remark}
\hfill
\begin{enumerate}[(i)]
\item
While it is not clear from the definition that condition~(iii) of
Theorem~\ref{thm:sa} is decidable, it is known to be equivalent to a decidable
condition. 
Briefly, let $\Gamma'$ be a core of $\Gamma$ defined on $D'\subseteq D$.
By~\cite[Lemma~3.7]{tz17:sicomp}, $\Gamma$ satisfies the BWC if and only if
$\Gamma'\cup\mathcal{C}_{D'}$ satisfies the BWC.
By~\cite[Theorem~2.8]{Kozik15:au}, $\Gamma'\cup\mathcal{C}_{D'}$ satisfies the
BWC if and only there are a ternary WNU $f$ and a $4$-ary WNU
$g$ in $\supp(\Gamma'\cup\mathcal{C}_{D'})$ satisfying
$f(y,x,x)=g(y,x,x,x)$ for all $x,y \in D'$. Finally, checking for the existence of such
operations can be done using a linear program.

\item
It is possible to obtain a
solution to an instance $\inst$ of $\VCSP(\Gamma)$ from the optimal value of the
SA$(3)$-relaxation of $I$~\cite[Section~3.6]{tz17:sicomp}.
\item Theorem~\ref{thm:sa} says that if
$\supp(\Gamma)$ violates the BWC then $\VCSP(\Gamma)$ requires more than a
constant level of the Sherali-Adams LP hierarchy for exact solvability.
The proof in~\cite{tz17:sicomp} actually shows that %assuming the BWC is violated then
in this case
$\Omega(\sqrt{n})$ levels are required for exact solvability of $n$-variable instances of
$\VCSP(\Gamma)$.
\end{enumerate}
\end{remark}

\subsection{Lasserre SDP Hierarchy}\label{sec:lass}

Let $I$ be an instance of the VCSP with $\phi_\inst(x_1,\dots,x_n) = \sum_{i=1}^q
\phi_i(\tup{x}_i)$, $X_i \subseteq V = \{x_1, \dots, x_n\}$ and $\phi_i \colon
D^{\ar(\phi_i)} \to \qq$. 
For $\sigma_i \colon X_i \to D$ and $\sigma_j \colon X_j \to D$, if
$\Crestrict{\sigma_i}{X_i\cap X_j} = \Crestrict{\sigma_j}{X_i\cap X_j}$
then we write $\sigma_i \circ \sigma_j \colon (X_i \cup X_j) \to D$ for
the assignment defined by $\sigma_i \circ \sigma_j(x)=\sigma_i(x)$ for $x\in
X_i$ and $\sigma_i \circ \sigma_j(x)=\sigma_j(x)$ otherwise. 

Let $k$ be an integer with $k\geq\max_i(\ar(\phi_i))$.\footnote{It also makes
sense to consider relaxations with $k<\max_i(\ar(\phi_i))$, in particular for
positive (algorithmic) results, such as the implication $(iii)\Rightarrow (ii)$
in Theorem~\ref{thm:sa}. For our main (impossibility) result, we will be
interested in $k$ which is linear in the number of variables of $I$.} The
$k$th level of the Lasserre SDP hierarchy~\cite{Lasserre01:jo},
henceforth called the $\Las(k$)-relaxation of $I$, is given by the following semidefinite
program (we follow the presentation from~\cite{Tulsiani09:stoc}).
Ensure that for every subset (including the empty set) $X \subseteq V$ with $|X| \leq k$ there is
some constraint $\phi_i(\tup{x}_i)$ with $X_i=X$, possibly by adding null
constraints.
The vector variables of the $\Las(k)$-relaxation, given in Figure~\ref{fig:sdp},
are 
$\blambda_i(\sigma)\in \mathbb{R}^t$ for every $i \in \left[q\right]$ and
assignment $\sigma \colon X_i \to D$. Here $t$ is the dimension of the real
vector space.\footnote{Typically, $t=(nd)^{O(k)}$ for an instance with $n$
variables over a domain of size $d$.}
We write $\blambda_0$ as a shorthand for $\blambda_i(\emptyset)$
where $i$ is the index for which $X_i=\emptyset$.

\begin{figure*}[t] 
\begin{alignat}{2}
\text{minimise}\ & \sum_{i = 1}^q \sum_{\sigma \in \feas(\phi_i)}||\blambda_{i}(\sigma)||^2 \phi_i(\sigma(\tup{x}_i)) \nonumber\\
\text{subject to}\ & & \nonumber\\
& ||\blambda_0|| = 1 \label{las:1} \tag{L1} \\
& \langle \blambda_i(\sigma_i), \blambda_j(\sigma_j) \rangle \geq 0 &\ &\
\forall i,j \in \left[q\right], \sigma_i \colon X_i \to D, \sigma_j \colon X_j \to D~\label{las:nonnegative} \tag{L2} \\
& ||\blambda_{i}(\sigma)||^2 = 0 &\ &\ \forall i \in \left[q\right], \sigma
\colon X_i \to D, \sigma(\tup{x}_i) \not\in \feas(\phi_i) \label{las:infeas} \tag{L3} \\
& \sum_{a\in D}||\blambda_{i}(a)||^2 = 1 &\ &\ \forall i \mbox{ with } |X_i|=1 \label{las:var} \tag{L4} \\
& \langle \blambda_i(\sigma_i), \blambda_j(\sigma_j) \rangle = 0 &\ &\ \forall i,j \in
\left[q\right], \sigma_i \colon X_i \to D, \sigma_j \colon X_j \to D \label{las:cons} \tag{L5} \\
& &\ &\ \Crestrict{\sigma_i}{X_i\cap X_j} \neq \Crestrict{\sigma_j}{X_i \cap X_j} \nonumber \\ 
& \langle \blambda_i(\sigma_i), \blambda_j(\sigma_j) \rangle = \langle
\blambda_{i'}(\sigma_{i'}), \blambda_{j'}(\sigma_{j'}) \rangle &\ &\ \forall
i,j,i',j' \in \left[q\right], X_i \cup X_j = X_{i'} \cup X_{j'} \label{las:split} \tag{L6} \\
& &\ &\ \sigma_i \colon X_i \to D, \sigma_j \colon X_j \to D, \sigma_{i'} \colon X_{i'} \to D \nonumber \\
& &\ &\ \sigma_{j'} \colon X_{j'} \to D, \sigma_i \circ \sigma_j = \sigma_{i'} \circ \sigma_{j'} \nonumber
\end{alignat}
\vspace*{-1em}
\caption{The $k$th level of the Lasserre SDP hierarchy, $\Las(k)$.}\label{fig:sdp}
\end{figure*}

For any fixed $k$ and any $t$ polynomial in the size of $\inst$, the $\Las(k)$-relaxation of $\inst$ is of
polynomial size in terms of $\inst$ and can be solved in polynomial
time~\cite{Gartner2012approximation}.\footnote{
Under technical assumptions which are satisfied by the Lasserre relaxation, 
SDPs can be solved approximately;
for any $\epsilon$ there is an algorithm that given an SDP returns
vectors for which the objective function is at most $\epsilon$ away from the
optimum value and the running time is polynomial in the input size and
$\log(1/\epsilon)$~\cite{Vandenberghe96:siam,Gartner2012approximation}. 
For any language $\Gamma$ of \emph{finite} size there is $\epsilon=\epsilon(\Gamma)$ such that solving the
SDP up to an additive error of $\epsilon$ suffices for exact solvability. For instance, take
$\epsilon$ such that
$\epsilon<\min_{\phi\in\Gamma}\min_{\vec{x},\vec{y}\in\feas(\phi),\phi(\vec{x})\neq\phi(\vec{y})}|\phi(\vec{x})-\phi(\vec{y})|$.
Since this paper deals with \emph{impossibility} results these matters are not
relevant but we mention it here for completeness.}
Note that $k$ may not necessarily be constant but it could depend on $n$, the number
of variables of $I$.

We write $\sdpval(\inst,\blambda,k)$ for the value of the SDP-solution $\blambda$ to the
$\Las(k)$-relaxation of $\inst$, and $\sdpopt{}{}(\inst,k)$ for its optimal value.

\begin{definition}
Let $\Gamma$ be a general-valued constraint language.
We say that $\VCSP(\Gamma)$ 
is \emph{solved by the $k$th} \emph{level of the Lasserre SDP hierarchy}
if for every instance $\inst$ of $\VCSP(\Gamma)$ we have
$\vcspopt(\inst)=\sdpopt{}{}(\inst,k)$.
\end{definition}

We say that an instance $\inst$ of $\VCSP(\Gamma)$ is a \emph{gap instance for}
the $k$th level of the Lasserre SDP hierarchy if
$\sdpopt{}{}(\inst,k)<\vcspopt(\inst)$.

\begin{definition}\label{def:linear}
Let $\Gamma$ be a general-valued constraint language.
We say that $\VCSP(\Gamma)$ \emph{requires linear levels} of the Lasserre SDP hierarchy
if there is a constant $0<c<1$ such that for all sufficiently large $n$ there
is an $n$-variable gap instance $\inst_n$ of $\VCSP(\Gamma)$ for $\Las(\lfloor
cn\rfloor)$.
\end{definition}

\subsection{Main Results}
\label{subsec:main}

Let $\mathcal{G}$ be an Abelian group over a finite set $G$ and let $r \geq 1$ be an integer.
Denote by $E_{\mathcal{G},r}$ the crisp constraint language over domain $G$ with,
for every $a \in G$, and $1 \leq m \leq r$, a relation 
$R^m_a = \{ (x_1, \dots, x_m) \in G^m \mid x_1 + \dots + x_m = a \}$.

We are now ready to state our main results.

\begin{theorem}\label{thm:main}
Let $\Gamma$ be a general-valued constraint language of finite size. The
following are equivalent:
\begin{enumerate}[(i)]
\item $\VCSP(\Gamma)$ requires linear levels of the Lasserre SDP hierarchy. \label{cndmain:las}
\item $\Gamma$ can simulate $E_{\mathcal{G},3}$ for some non-trivial Abelian group $\mathcal{G}$. \label{cndmain:equations}
\item $\supp(\Gamma)$ violates the BWC.  \label{cndmain:BWC}
\end{enumerate}
\end{theorem}

Theorems~\ref{thm:sa} and~\ref{thm:main} give the following.

\begin{corollary}
\label{cor:alg-dich}
Let $\Gamma$ be a general-valued constraint language of finite size. Then, either
$\VCSP(\Gamma)$ is solved by the third level of the Sherali-Adams LP relaxation,
or $\VCSP(\Gamma)$ requires linear levels of the Lasserre SDP relaxation.
\end{corollary}
\begin{proof}
Either $\supp(\Gamma)$ satisfies the BWC, in which case $\VCSP(\Gamma)$ is solved by the
third level of the Sherali-Adams LP relaxation by Theorem~\ref{thm:sa}, or
$\supp(\Gamma)$ violates the BWC, in which case $\VCSP(\Gamma)$ requires linear
levels of the Lasserre SDP hierarchy by Theorem~\ref{thm:main}.
\end{proof}

Recall that a constraint language $\Gamma$ is called \emph{crisp} if it contains
only (unweighted) relations. Our result covers this special case, and thus we
get the following corollary, which was independently obtained (using a different
proof) in~\cite{Atserias17:icalp}.

\begin{corollary}
Let $\Gamma$ be a crisp constraint language of finite size. Then, either
$\VCSP(\Gamma)$ is solved by the third level of the Sherali-Adams LP relaxation,
or $\VCSP(\Gamma)$ requires linear levels of the Lasserre SDP relaxation.
\end{corollary}

A constraint language $\Gamma$ is called \emph{finite-valued}~\cite{tz16:jacm}
if for every $\phi\in\Gamma$ it holds $\phi(\vec{x})<\infty$ for every
$\vec{x}$. In this special case, we get the following result, which was
independently obtained (using a different proof) in~\cite{Dawar17:lics}.

\begin{corollary}
Let $\Gamma$ be a finite-valued constraint language of finite size. Then, either
$\VCSP(\Gamma)$ is solved by the first level of the Sherali-Adams LP relaxation,
or $\VCSP(\Gamma)$ requires linear levels of the Lasserre SDP relaxation.
\end{corollary}

\begin{proof}
Let $D$ be the domain of $\Gamma$. 
If $\VCSP(\Gamma)$ is \emph{not} solved by the first level of the Sherali-Adams
LP relaxation, then~\cite{tz16:jacm} shows (in different terminology) that 
$\Gamma$ can simulate $\phi_{\sf mc}$ (cf. Example~\ref{ex:langs}). Using
$\phi_{\sf mc}$ together with the unary constant relations $c_0$ and $c_1$,
it is then not difficult to express a ternary weighted relation $\phi$ such 
that $\phi(x,y,z)$ minimises on $x+y+z = 0 \pmod{2}$.
Now, $R^3_0 = \opt(\phi)$ together with $c_0$ and $c_1$
can express all remaining relations in $E_{\mathbb{Z}_2,3}$.
Overall, we conclude that $\Gamma$ can simulate 
$E_{\mathbb{Z}_2,3}$, which proves the claim by~Theorem~\ref{thm:main}.
\end{proof}

Lee et al.~\cite{Lee14:arxiv-sdp,Lee15:stoc} give some very strong results on
approximation-preserving reductions between SDP relaxations. They give a general
reduction turning lower bounds on the number of levels of the Lasserre SDP
hierarchy needed for approximation to lower bounds on the size of \emph{arbitrary} SDP
relaxations. In particular, they show that if linear levels of the Lasserre SDP
relaxation are required for some problems then no \emph{polynomial-size SDP
relaxation} suffices.
We now briefly discuss how their result together with
Theorem~\ref{thm:main} can be used to derive the same consequence for
$\VCSP(\Gamma)$ when $\supp(\Gamma)$ violates the BWC. 

Lee et al. give in~\cite[Theorem~6.4]{Lee14:arxiv-sdp} a reduction for turning lower bounds
on the number of levels of the Lasserre SDP hierarchy needed for
approximate maximisation of Max-CSPs to lower bounds on the size of
arbitrary SDP relaxations. In order to apply their theorem in our setting, a
number of differences in the setup of this paper and~\cite{Lee14:arxiv-sdp} must
be addressed. First,
\cite[Theorem~6.4]{Lee14:arxiv-sdp} is stated only for Boolean domains
and proved using~\cite[Theorem~3.8]{Lee14:arxiv-sdp}. However, a
generalisation to arbitrary fixed finite domains follows
from~\cite[Theorem~7.2]{Lee14:arxiv-sdp}~\cite{RS}. Second, the results
in~\cite{Lee14:arxiv-sdp,Lee15:stoc} are formulated for the
sum-of-squares SDP hierarchy, which 
is equivalent to the Lasserre SDP hierarchy: the $k$th level of the
sums-of-squares SDP hierarchy is the same as the $(k/2)$th level of the
Lasserre SDP hierarchy. Third, while the results
in~\cite{Lee14:arxiv-sdp,Lee15:stoc} are formulated for constraint
languages consisting of a single $\{0,1\}$-valued weighted relation, the
proofs give the same result for constraint languages (of finite size)
consisting of $[0,1]$-valued weighted relations of different
arities~\cite{RS}. Finally, while the work
in~\cite{Lee14:arxiv-sdp,Lee15:stoc} deals with maximisation problems,
for exact solvability we can equivalently turn to minimisation problems.

\subsection{Proof of Theorem~\ref{thm:main}}
\label{subsec:overview}

Let $\Gamma$ be a general-valued constraint language of finite size. If
$\supp(\Gamma)$ violates the BWC then we aim to prove that $\VCSP(\Gamma)$
requires linear levels of the Lasserre SDP hierarchy. 

We will follow the approach used in~\cite{tz17:sicomp} to prove the implication
$(\ref{cnd:bound})\Longrightarrow (\ref{cnd:BWC})$ of Theorem~\ref{thm:sa}.
This is
based on the idea that if $\supp(\Gamma)$ violates the BWC, then $\Gamma$ can
simulate linear equations in some Abelian group. In order to establish the
implications
$(\ref{cndmain:BWC})\Longrightarrow(\ref{cndmain:equations})\Longrightarrow(\ref{cndmain:las})$
of Theorem~\ref{thm:main}, it suffices to show that linear equations require
linear levels of the Lasserre SDP hierarchy and that the simulation preserves
exact solvability by the Lasserre SDP hierarchy (up to a constant factor in the
level of the hierarchy). Our contribution is proving the latter. The former is
known~\cite{Grigoriev01:tcs,Schoenebeck08:focs,Tulsiani09:stoc}, as we will now
discuss.

\begin{theorem}[\cite{Chan16:jacm}]\label{thm:eq3linearlass}
Let $\mathcal{G}$ be a finite non-trivial Abelian group.
Then, $\VCSP(E_{\mathcal{G},3})$ requires linear levels of the Lasserre
SDP hierarchy.
\end{theorem}

For Abelian groups of prime orders, Tulsiani showed that there is a
constant $0<c<1$ such that for every large enough $n$ there is an instance
$\inst_n$ of $\VCSP(E_{\mathcal{G},3})$ on $n$ variables with
$\vcspopt(\inst_n)=\infty$ and $\sdpopt{}{}(\inst_n,\lfloor cn\rfloor)=0$; i.e.,
$\inst_n$ is a gap instance for $\Las(\lfloor
cn\rfloor)$~\cite[Theorem~4.2]{Tulsiani09:stoc}.\footnote{We note that
\cite{Tulsiani09:stoc} uses different terminology from ours: Max-CSP($P$) for a
$k$-ary predicate $P$ applied to literals rather than variables.} This work was
based on the result of Schoenebeck who showed it for Boolean
domains~\cite{Schoenebeck08:focs}, 
thus rediscovering the work of Grigoriev~\cite{Grigoriev01:tcs}. A
generalisation to all Abelian groups was then established by Chan
in~\cite[Appendix~D]{Chan16:jacm}. 
Theorem~\ref{thm:eq3linearlass} states that 
distinguishing satisfiable instances of
$\VCSP(E_{\mathcal{G},3})$ from instances in which not all constraints are
simultaneously satisfiable requires linear levels of the Lasserre SDP hierarchy. We remark that the results
in~\cite{Schoenebeck08:focs,Tulsiani09:stoc,Chan16:jacm} actually prove
something much stronger: 
even distinguishing satisfiable instances from instances in which only a small
fraction of the constraints are simultaneously satisfiable requires linear
levels of the Lasserre SDP hierarchy.

The following notion of reduction is key in this paper.

\begin{definition}\label{def:reducesto}
Let $\Gamma$ and $\Delta$ be two general-valued constraint languages of finite size. We
write $\Delta \reducesto \Gamma$ if there is a polynomial-time reduction from
$\VCSP(\Delta)$ to $\VCSP(\Gamma)$ with the following property: there is a constant $c\geq 1$ depending only on $\Gamma$
and $\Delta$ such that for any $k\geq 1$, if $\Las(k)$ solves $\VCSP(\Gamma)$
then $\Las(ck)$ solves $\VCSP(\Delta)$.
\end{definition}

By Definition~\ref{def:reducesto}, $\reducesto$ reductions compose. Let
$\Delta\reducesto\Gamma$. By Definitions~\ref{def:linear}
and~\ref{def:reducesto}, if $\VCSP(\Delta)$ requires linear levels of the
Lasserre SDP hierarchy then so does $\VCSP(\Gamma)$.
An analogous notion of reduction for the Sherali-Adams LP hierarchy, $\sa$, was
used in~\cite{tz17:sicomp}.

The following theorem is the main technical contribution of the paper. It shows
that a general-valued constraint language can be augmented with various additional
weighted relations while preserving exact solvability in the Lasserre SDP
hierarchy up to a constant factor in the level of the hierarchy. It is a strengthening of
Theorem~\cite[Theorem~5.5]{tz17:sicomp}, which showed that the same
additional weighted relations preserve exact solvability in the Sherali-Adams LP
hierarchy.

\begin{theorem}\label{thm:reductions}
Let $\Gamma$ be a general-valued constraint language of finite size on domain $D$.
The following holds:
\begin{enumerate}
\item\label{red:express}
If $\phi$ is expressible in $\Gamma$, then $\Gamma \cup \{\phi\} \reducesto \Gamma$.
\item\label{red:equality}
$\Gamma \cup \{\eq{D}\} \reducesto \Gamma$.
\item\label{red:interpret}
If $\Gamma$ interprets the general-valued constraint language $\Delta$ of finite size, then $\Delta \reducesto \Gamma$.
\item\label{red:feasopt}
If $\phi \in \Gamma$, then 
$\Gamma \cup \{\opt(\phi)\} \reducesto \Gamma$ and
$\Gamma \cup \{\feas(\phi)\} \reducesto \Gamma$.
\item\label{red:coreconstants} If $\Gamma'$ is a core of $\Gamma$ on domain $D' \subseteq D$, then $\Gamma' \cup \mathcal{C}_{D'} \reducesto \Gamma$.
\end{enumerate}
\end{theorem}

\begin{proof} 
The proof is to a large extent based on a technical lemma,
Lemma~\ref{lem:las-reduc}, which is stated and proved in
Section~\ref{sec:reductions}. This lemma shows that, subject to some
consistency conditions, a polynomial-time reduction between two
constraint languages $\Delta$ and $\Gamma$ that is based on locally
replacing valued constraints with weighted relations in $\Delta$ by
gadgets expressed in $\Gamma$ can be turned into an
$\reducesto$-reduction.
The same approach was used in~\cite[Theorem~5.5]{tz17:sicomp} for constructing
$\sa$-reductions for (\ref{red:express}--\ref{red:interpret}), and
(\ref{red:coreconstants}). In these cases, it therefore essentially
suffices to replace the applications of~\cite[Lemma~6.1]{tz17:sicomp} by
applications of Lemma~\ref{lem:las-reduc} in the proofs
of~\cite[Lemmas~6.2--6.4, and 6.7]{tz17:sicomp}.

For case (\ref{red:interpret}), we remark that our definition differs slightly from that of~\cite{tz17:sicomp} in that we incorporate applications of the operations $\opt$ and $\feas$ as well as
scaling by nonnegative rational constants and addition of rational constants
in the definition of $\langle \Gamma \rangle$.
To accommodate for the operations $\opt$ and $\feas$ in the proof,
it suffices to add an application of (\ref{red:feasopt}).
Furthermore, scaling can be
implemented by repeated constraints and the addition of a constant
changes the value of the objective function of the VCSP instance by the
same constant as the objective function of the SDP relaxation, for all
feasible solutions to the corresponding problems.

For case (\ref{red:coreconstants}), the proof in~\cite[Lemmas~6.7]{tz17:sicomp} also
refers to~\cite[Lemma~5.6]{tz17:sicomp} which also hold
for $\reducesto$-reductions by Lemma~\ref{lem:killingf} below, and cases~(\ref{red:express}) and~(\ref{red:feasopt}).

The remaining two reductions in (\ref{red:feasopt}) are shown in a more straightforward
way for $\sa$-reductions in~\cite[Lemmas~6.5 and 6.6]{tz17:sicomp}.
Here, we argue that the proof of~\cite[Lemmas~6.5]{tz17:sicomp} goes through 
for $\reducesto$-reductions as well, which shows that $\Gamma \cup \{\opt(\phi)\} \reducesto \Gamma$.
We omit the analogous argument for the reduction $\Gamma \cup \{\feas(\phi)\} \reducesto \Gamma$.
In the proof of~\cite[Lemmas~6.5]{tz17:sicomp},
an instance $I$ of $\VCSP(\Gamma \cup \{\opt(\phi)\})$ is transformed into an instance $J$ of
$\VCSP(\Gamma)$ by replacing all occurrences of  $\opt(\phi)$ by multiple copies of $\phi$.
It is then shown that
if $I$ is a gap instance for the SA$(k)$-relaxation,
and $\lambda$ is an optimal solution to this relaxation, 
then $\lambda$ is also a solution to the SA$(k)$-relaxation of $J$.
Moreover, $\lambda$ attains a better value than $\vcspopt(J)$, hence $J$ is also
a gap instance. This argument goes through also if we take $I$ to be a gap
instance for the $\Las(k)$-relaxation, and $\blambda$ an optimal solution to this relaxation.
The exact same solution $\blambda$ then also shows that $J$ is a gap instance
for the $\Las(k)$-relaxation.
\end{proof}

In order to finish the proof of Theorem~\ref{thm:main}, we need a few additional
results. The following result follows, as described in the proof
of~\cite[Theorem~5.4]{tz17:sicomp}, from~\cite{Atseriasetal09:tcs,Kozik15:au}.

\begin{theorem}[\protect{\cite[Theorem~5.4]{tz17:sicomp}}]\label{thm:Agroup}
Let $\Delta$ be a crisp constraint language of finite size
that contains all constant unary relations.
If $\pol(\Delta)$ violates the BWC, then there exists a finite non-trivial Abelian 
group $\mathcal{G}$
such that $\Delta$ interprets $E_{\mathcal{G},r}$, for every $r \geq 1$.
\end{theorem}

The following two lemmas, together with cases (\ref{red:express}) and
(\ref{red:feasopt}) of Theorem~\ref{thm:reductions}, extend~\cite[Lemma~5.6 and
Lemma~5.7]{tz17:sicomp} from $\sa$-reductions to $\reducesto$-reductions.

\begin{lemma}\label{lem:killingf}
Let $\Gamma$ be a general-valued constraint language over domain $D$ and let $F$ be a
set of operations over $D$.
If $\supp(\Gamma) \cap F = \emptyset$, then there exists a crisp constraint
language $\Delta \subseteq \langle \Gamma \rangle$ such that
$\pol(\Delta) \cap F = \emptyset$. Moreover, if $\Gamma$ and $F$ are finite then so is
$\Delta$.
\end{lemma}

\begin{proof}
By~\cite[Lemma~2.9]{tz17:sicomp}, for each $f \in F \cap \pol(\Gamma)$,
there is an instance $I_f$ of VCSP$(\Gamma)$
such that $f \not\in \pol(\opt(\phi_{I_f}))$.
Let $\Delta = \{ \opt(\phi_{I_f}) \mid f \in F \} \cup \{ \feas(\phi) \mid \phi
\in \Gamma \} \subseteq \langle \Gamma \rangle$.
For $f \in F \cap \pol(\Gamma)$, we have $f \not\in \pol(\opt(\phi_{I_f})) \supseteq \pol(\Delta)$.
For $f \in F \setminus \pol(\Gamma)$, we have $f \not\in \pol(\phi)$, for some $\phi \in \Gamma$,
so $f \not\in \pol(\Delta)$.
It follows that $\pol(\Delta) \cap F = \emptyset$.
\end{proof}

\begin{lemma}\label{lem:crispwnus}
Let $\Gamma$ be a general-valued constraint language of finite size.
If $\supp(\Gamma)$ violates the BWC,
then there is a crisp constraint language $\Delta \subseteq \langle \Gamma
\rangle$ of finite size 
such that $\pol(\Delta)$ violates the BWC.
\end{lemma}

\begin{proof}
Since $\supp(\Gamma)$ violates the BWC, there exists an $m \geq 3$ such that
$\supp(\Gamma)$ does not contain any $m$-ary WNU.
Let $F$ be the (finite) set of all $m$-ary WNUs.
The result follows by applying Lemma~\ref{lem:killingf}
to $\Gamma$ and $F$.
\end{proof}

We are now ready to prove Theorem~\ref{thm:main}.

\begin{proof}[Proof of Theorem~\ref{thm:main}] 
Theorem~\ref{thm:sa} gives the implication
$(\ref{cndmain:las})\Longrightarrow(\ref{cndmain:BWC})$ by contraposition: if
$\supp(\Gamma)$ satisfies the BWC then, by Theorem~\ref{thm:sa}, $\VCSP(\Gamma)$
is solved by any constant level $k$ of the Sherali-Adams LP hierarchy with $k
\geq 3$, and thus also by the $k$th level of the Lasserre SDP hierarchy for $k
\geq \ar(\Gamma)$.

Now, suppose that $\supp(\Gamma)$ violates the BWC.
Let $\Gamma'$ be a core of $\Gamma$ on a domain $D' \subseteq D$ and
let $\Gamma_c = \Gamma' \cup \mathcal{C}_{D'}$.
By~\cite[Lemma~3.7]{tz17:sicomp}, $\supp(\Gamma_c)$
also violates the BWC.
By Lemma~\ref{lem:crispwnus},
there exists a finite crisp constraint language $\Delta$ such that $\Delta$  has an
interpretation in $\Gamma_c$ and $\pol(\Delta)$ violates the BWC.
Since $\mathcal{C}_{D} \subseteq \Gamma_c$, we may assume,
without loss of generality, that $\mathcal{C}_{D} \subseteq \Delta$.
By Theorem~\ref{thm:Agroup}, there exists a finite non-trivial Abelian group
$\mathcal{G}$ and an interpretation of $E_{\mathcal{G},3}$ in $\Delta$.
Since interpretations compose, 
$E_{\mathcal{G},3}$ has an interpretation in $\Gamma_c$.
Therefore, $\Gamma$ can simulate $E_{\mathcal{G},3}$ which
gives the implication $(\ref{cndmain:BWC})\Longrightarrow(\ref{cndmain:equations})$.

Finally, by Theorem~\ref{thm:eq3linearlass}, 
$\VCSP(E_{\mathcal{G},3})$ requires linear levels of the Lasserre SDP hierarchy.
By Theorem~\ref{thm:reductions}(\ref{red:interpret}) and (\ref{red:coreconstants}),
we have $E_{\mathcal{G},3} \reducesto \Gamma_c \reducesto \Gamma$.
Consequently, $\VCSP(\Gamma)$ requires linear levels of the Lasserre SDP
hierarchy as well.
This gives the implication
$(\ref{cndmain:equations})\Longrightarrow(\ref{cndmain:las})$.
\end{proof}

\section{An $\reducesto$-Reduction Scheme} 
\label{sec:reductions}

In this section, we will prove Lemma~\ref{lem:las-reduc}, which is the key
technique used to establish cases (1)--(3) and (5) of
Theorem~\ref{thm:reductions}. It is an analogue
of~\cite[Lemma~6.1]{tz17:sicomp}, which does the same for the
$\sa$-reductions,
and the proof is closely modelled on that
of~\cite[Lemma~6.1]{tz17:sicomp}.

The following observation will be used throughout this section: since the
set of vectors $\{ \blambda_i(\tau) \mid \tau \in D^{X_i} \}$ for a feasible
solution $\blambda$ is orthogonal by~(\ref{las:cons}), it follows that $\|
\sum_{\tau \in T} \blambda_i(\tau) \|^2 = \sum_{\tau \in T} \langle
\blambda_i(\tau), \blambda_i(\tau) \rangle$ for any subset $T \subseteq
D^{X_i}$.

We will also use the following lemma which can be seen as
an additional set of constraints on the $\Las(k)$-relaxation but which
follows directly from the others.

\begin{lemma}\label{lem:vec-marg}
Every feasible solution $\blambda$ to the $\Las(k)$-relaxation satisfies, in
addition to (\ref{las:1})--(\ref{las:split}), the following:
\begin{multline}\label{las:vec-marg}\tag{L7}
\sum_{\tau \colon \Crestrict{\tau}{X_j} = \sigma}
\blambda_i(\tau)  = \blambda_j(\sigma)\hfill
%\\
\forall i,j\in [q], X_j \subseteq X_i, |X_i| \leq k, \sigma \colon X_j \to D.
\end{multline}
\end{lemma}

\begin{proof}
Consider the norm of the vector $\sum_{\tau \colon
\Crestrict{\tau}{X_j} = \sigma} \blambda_i(\tau) - \blambda_j(\sigma)$.

\begin{equation*}
\begin{aligned}
&\hspace*{1.3em} \| \sum_{\tau \colon \Crestrict{\tau}{X_j} = \sigma} \blambda_i(\tau) - \blambda_j(\sigma) \|^2 \\
&=
\| \sum_{\tau \colon \Crestrict{\tau}{X_j} = \sigma} \blambda_i(\tau) \|^2 - 2
\langle \sum_{\tau \colon X_i \to D} \blambda_i(\tau), \blambda_j(\sigma)
\rangle + \| \blambda_j(\sigma) \|^2 \\
&=
\| \sum_{\tau \colon \Crestrict{\tau}{X_j} = \sigma} \blambda_i(\tau) \|^2 - 2
\sum_{\tau \colon X_i \to D} \langle \blambda_i(\tau), \blambda_j(\sigma)
\rangle + \| \blambda_j(\sigma) \|^2 \\
&=
\| \sum_{\tau \colon \Crestrict{\tau}{X_j} = \sigma} \blambda_i(\tau) \|^2 - 2
\sum_{\tau \colon X_i \to D} \langle \blambda_i(\tau), \blambda_i(\tau) \rangle
+ \| \blambda_j(\sigma) \|^2 \\
&=
- \| \sum_{\tau \colon \Crestrict{\tau}{X_j} = \sigma} \blambda_i(\tau) \|^2 +
\| \blambda_j(\sigma) \|^2,
\end{aligned}
\end{equation*}
where the next to last equality follows from (\ref{las:split}) since $X_j
\subseteq X_i$ and $\sigma = \Crestrict{\tau}{X_j}$.
We see that the equality in the lemma is equivalent to:
\begin{equation}
\label{equiv}
\| \sum_{\tau \colon \Crestrict{\tau}{X_j} = \sigma} \blambda_i(\tau) \|^2 = \|
\blambda_j(\sigma) \|^2.
\end{equation}

We finish the proof by induction on $|X_i \setminus X_j| \geq 1$.
There are two base cases:
\begin{enumerate}[(i)]
\item\label{base1}
If
$|X_i \setminus X_j| = 1$ and $X_j = \emptyset$,
then
(\ref{equiv}) follows immediately from (\ref{las:1}) and (\ref{las:var}).
\item
If
$|X_i \setminus X_j| = 1$ and $X_j \neq \emptyset$,
then
let $X_r = \{x\} = X_i \setminus X_j$ be a scope on the single variable $x$,
and, for $a \in D$, let $\sigma_a$ be the assignment $\sigma_a(x) = a$.
Now, (\ref{equiv}) follows from:
\begin{equation*}
\begin{aligned}
\| \sum_{\tau \colon \Crestrict{\tau}{X_j} = \sigma} \blambda_i(\tau) \|^2
& \makebox[2em]{=}
\sum_{a \in D} \langle \blambda_i(\sigma_a \circ \sigma), \blambda_i(\sigma_a
\circ \sigma) \rangle \\
& \makebox[2em]{$\stackrel{\text{(\ref{las:split})}}{=}$}
\sum_{a \in D} \langle \blambda_{r}(\sigma_a), \blambda_j(\sigma) \rangle \\
& \makebox[2em]{=}
\langle \sum_{a \in D} \blambda_{r}(\sigma_a), \blambda_j(\sigma) \rangle \\
& \makebox[2em]{$\stackrel{\text{(\ref{base1})}}{=}$}
\langle \blambda_0, \blambda_j(\sigma) \rangle, \\
& \makebox[2em]{$\stackrel{\text{(\ref{las:split})}}{=}$}
\langle \blambda_j(\sigma), \blambda_j(\sigma) \rangle.
\end{aligned}
\end{equation*}
\end{enumerate}

Finally, assume that $|X_i \setminus X_j| > 1$ and that $x \in X_i \setminus
X_j$.
Let $r$ be an index such that $X_r = X_j \cup \{x\}$, and,
for $a \in D$, let $\sigma_a$ be the assignment $\sigma_a(x) = a$.
Then,
\begin{equation*}
\begin{aligned}
\sum_{\tau \colon \Crestrict{\tau}{X_j} = \sigma} \blambda_i(\tau)
&=
\sum_{a \in D} \sum_{\tau \colon \Crestrict{\tau}{X_r} = \sigma \circ \sigma_a}
\blambda_i(\tau) \\
&=
\sum_{a \in D} \blambda_r(\sigma \circ \sigma_a) \\
&=
\blambda_j(\sigma),
\end{aligned}
\end{equation*}
where the last two equalities follow by induction.
\end{proof}

For a solution $\blambda$ to the $\Las(k)$-relaxation of $\inst$ with the
objective function $\sum_{i=1}^q\phi(\vec{x}_i)$, we denote by $\bsupp(\blambda_i)$
the positive support of $\blambda_i$, i.e., $\bsupp(\blambda_i)=\{\sigma \colon X_i
\to D \mid ||\blambda_i(\sigma)||^2>0\}$. 

The following technical lemma is the basis for the reductions in Theorem~\ref{thm:reductions}.

\begin{lemma}\label{lem:las-reduc}
Let $\Delta$ and $\Delta'$ be general-valued constraint languages of finite size over domains $D$ and $D'$,
respectively.

Let $(I, i) \mapsto J_i$ be a map 
that to each instance $I$
of $\VCSP(\Delta)$ with variables $V$ and objective function $\sum_{i=1}^q \phi_i(\tup{x}_i)$,
and index $i \in [q]$,
associates an instance $J_i$ of $\VCSP(\Delta')$ with variables $Y_i$
and objective function $\phi_{J_i}$.
Let $J$ be the $\VCSP(\Delta')$ instance with variables $V' = \bigcup_{i=1}^q Y_i$ and
objective function $\sum_{i=1}^q \phi_{J_i}$.

Suppose that the following holds:
\begin{enumerate}[(a)]
\item\label{cond:jtoi}
For every satisfying and optimal assignment $\alpha$ of $J$, there exists a satisfying assignment $\sigma^{\alpha}$
of $I$ such that
\[
\vcspval(I,\sigma^{\alpha}) \leq \vcspval(J,\alpha).
\]
\end{enumerate}

Furthermore, suppose that for any $k \geq \ar(\Delta)$, and any feasible solution
$\blambda$ of the $\Las(k)$-relaxation of $I$, the following properties hold:
\begin{enumerate}[(a)]
\setcounter{enumi}{1}
\item\label{cond:itoj}
For $i \in [q]$, and $\sigma \colon X_i \to D$ with positive support in $\blambda$,
there exists a satisfying assignment $\alpha^{\sigma}_i$ of $J_i$ such that
\[
\phi_i(\sigma(\tup{x}_i)) \geq \vcspval(J_i, \alpha^{\sigma}_i);
\]
\item\label{cond:consistent}
for $i, r \in [q]$, any $X\subseteq V$ with $X_i \cup X_r \subseteq X$, and $\sigma \colon X
\to D$ with positive support in $\blambda$,
\[
\alpha^{\sigma_i}_i|_{Y_i \cap Y_r} = 
\alpha^{\sigma_r}_r|_{Y_i \cap Y_r},
\]
where
$\sigma_i = \Crestrict{\sigma}{X_i}$ and $\sigma_r = \Crestrict{\sigma}{X_r}$.
\end{enumerate}

Then, $I \mapsto J$ is a many-one reduction from $\VCSP(\Delta)$ to
$\VCSP(\Delta')$ that certifies $\Delta \reducesto \Delta'$.
\end{lemma}

\begin{proof}
First, we show that $\vcspopt(I) = \vcspopt(J)$.
From condition (\ref{cond:jtoi}), if $J$ is satisfiable, then so is $I$ and 
$\vcspopt(I) \leq \vcspopt(J)$.
Conversely, if $I$ is satisfiable, and $\sigma$ is an optimal assignment to $I$, then
the $\Las(2k)$
solution $\blambda$,
where $k \geq \ar(\Delta)$,
that assigns a fixed unit vector to $\Crestrict{\sigma}{X}$ for every
$X \subseteq V$ with $|X| \leq 2k$ is feasible.
Let $\sigma_i = \Crestrict{\sigma}{X_i}$.
By (\ref{cond:itoj}), there exist 
satisfying assignments $\alpha_i^{\sigma_i}$ of $J_i$, for all $i \in [q]$, such that
$\vcspopt(I) \geq \sdpopt{}{}(I,2k) \geq \sum_{i \in [q]} \vcspval(J_i,\alpha_i^{\sigma_i})$.
Define an assignment $\alpha \colon V' \to D'$ by letting
$\alpha(y) = \alpha_i^{\sigma_i}(y)$ for an arbitrary $i$ such that $y \in Y_i$.
We claim that $\Crestrict{\alpha}{Y_i} = \alpha_i^{\sigma_i}$, for all $i\in
[q]$. From this it follows that $\alpha$ is a satisfying assignment to $J$ such
that
$\sum_{i \in [q]} \vcspval(J_i,\alpha_i^{\sigma_i}) = \vcspval(J,\alpha) \geq \vcspopt(J)$,
and hence that $\vcspopt(I) \geq \vcspopt(J)$.
Indeed, let $y \in V'$ and assume that $y \in Y_{i}$ and $y \in Y_{r}$.
Let $X = X_i \cup X_r$.
Then, since $k\geq\ar(\Delta)$ and $||\blambda(\Crestrict{\sigma}{X})||^2>0$, it follows from (\ref{cond:consistent}) that $\alpha_i^{\sigma_i}(y) = \alpha_r^{\sigma_r}(y)$.
 
Let $k'$ be arbitrary and let $k = \max \{ k', \ar(\Delta') \} \cdot
\ar(\Delta)$. Assume that $I$ is a gap instance for the $\Las(2k)$-relaxation of
$\VCSP(\Delta)$, and let $\blambda$ be a feasible solution such that
$\sdpval(I,\blambda,2k) < \vcspopt(I)$ (where $\vcspopt(I)$ may be $\infty$, i.e.\
$I$ may be unsatisfiable). We show that there is a feasible solution $\bkappa$ to
the $\Las(k')$-relaxation of $J$ such that $\sdpval(J,\bkappa,k') \leq
\sdpval(I,\blambda,2k)$.\footnote{We remark here that the vectors in the feasible solution $\bkappa$ will live in the same
space $\mathbb{R}^t$ as those of $\blambda$. 
This is not a problem as long as $t$ is chosen sufficiently large enough for both of the relaxations.}
Then, by condition (\ref{cond:jtoi}), we have $\vcspopt(I) \leq \vcspopt(J)$.
Hence, $\sdpval(J,\bkappa,k') \leq \sdpval(I,\blambda,2k) < \vcspopt(I) \leq \vcspopt(J)$, so
$J$ is a gap instance for the $\Las(k')$-relaxation of $\VCSP(\Delta')$. Since
$k'$ was chosen arbitrarily, we have $\Delta \reducesto \Delta'$.

To this end,
augment $I$ with null constraints on $X_{q+1}, \dots, X_{q'}$ so that
for every at most $2k$-subset $X \subseteq V$, there exists an $i \in [q']$
such that $X_i = X$.
Rewrite the objective function of $J$ as $\sum_{j=1}^{p} \phi'_j(\tup{y}'_j)$,
$\phi' \in \Delta'$,
where,
by possibly first adding extra null constraints to $J$,
we will assume that for every at most $k'$-subset $Y \subseteq V'$, there exists a $j \in [p]$
such that $Y'_j = Y$.
Here, $Y'_j$ denotes the set of variables occurring in the tuple $\tup{y}'_j$.
For each $i \in [q]$, let $C_i$ be the set of indices $j \in [p]$ corresponding
to the valued constraints in the instance $J_i$.

For $m \geq 1$, define $X_{(\leq m)} = \{ X \subseteq V \mid X = \bigcup_{i \in S} X_i, S \subseteq [q], \left| X \right| \leq m \}$. This is the set of all scopes $X \subseteq V$ of size at most $m$ that can be written as a union of scopes $X_j$ with $j \in [q]$.
Note that this set includes some, but not necessarily all, of the scopes $X_i$, $i \in [q'] \setminus [q]$.

We now extend $\alpha^\sigma_i$ to all indices 
$i \in [q'] \setminus [q]$ for which $X_i \in X_{(\leq 2k)}$.
For a scope $X \in X_{(\leq 2k)}$, define $Y_X = \bigcup_{j \in [q] : X_j \subseteq X} Y_j$.
The idea is that an assignment $\sigma_i \colon X_i \to D$ with $X_i \in X_{(\leq 2k)}$ 
will be mapped to an assignment $\alpha^\sigma_i \colon Y_{X_i} \to D'$.
The assignment $\alpha^\sigma_i$ will be the union of the assignments
$\alpha^\sigma_j$ over all $j \in [q]$ that satisfy $X_j \subseteq X_i$.
For this to be well defined, we need to verify that the assignments $\alpha^\sigma_j$ are pairwise consistent:
Let $\sigma \in \bsupp(\blambda_i)$,
and $r, s \in [q]$ be such that $X_r \cup X_s \subseteq X_i$ and $y \in Y_r \cap Y_s$.
Then, by (\ref{cond:consistent}),
it holds that $\alpha^{\sigma_r}_r(y) = \alpha^{\sigma_s}_s(y)$.
Therefore, we can uniquely define $\alpha^{\sigma}_i \colon Y_{X_i} \to D'$ by letting
$\alpha^{\sigma}_i(y) = \alpha^{\sigma_r}_r(y)$ for any choice of $r \in [q]$ with
$X_r \subseteq X_i$ and $y \in Y_r$.
This definition is consistent with $\alpha^\sigma_i$ for $i \in [q]$ in the sense
that $(\ref{cond:consistent})$ now holds for all $i, r \in [q']$ such that $X_i, X_r \in X_{(\leq 2k)}$.

Let $j \in [p]$ and define $X_{(\leq m)}(Y'_j) = \{ X \in X_{(\leq m)} \mid Y'_j \subseteq Y_X \}$.
In particular, if $X_i \in X_{(\leq 2k)}(Y'_j)$, then $\alpha^\sigma_i$ as defined above can be restricted to an assignment on $Y'_j$.
Next, we show that $X_{(\leq 2k)}(Y'_j)$ is in fact non-empty so that such a scope $X_i$ always exists.
Let $n = |V|$.
The set $X_{(\leq n)}(Y'_j)$ is non-empty since $\bigcup_{i \in [q]} X_i \in X_{(\leq n)}(Y'_j)$.
Arbitrarily pick $X \in X_{(\leq n)}(Y'_j)$.
Then, $X = \bigcup_{i \in S} X_i$ for some $S \subseteq [q]$.
For each $y \in Y'_j$, let $i(y) \in S$ be an index such that $y \in Y_{i(y)}$
and let $X' = \bigcup_{y \in Y'_j} X_{i(y)}$.
Then, $Y'_j \subseteq Y_{X'}$, $X' \subseteq X$,
and $\left|X'\right| \leq \max\{k',\ar(\Delta')\} \cdot \ar(\Delta) = k$,
so $X' \in X_{(\leq k)}(Y'_j)$.
In other words,
\begin{equation}\label{eq:thindown}
\text{for every } X \in X_{(\leq n)}(Y'_j), \text{there exists } i\in[q'] \text{ such that } X_i \subseteq X \text{ and } X_i \in X_{(\leq k)}(Y'_j).
\end{equation}
In particular~(\ref{eq:thindown}) implies that $X_{(\leq 2k)}(Y'_j) \supseteq X_{(\leq k)}(Y'_j)$ is non-empty for every $j \in [p]$.

For $j \in [p]$, $\alpha \colon Y'_j \to D'$, and $i \in [q']$ such that $X_i \in X_{(\leq 2k)}(Y'_j)$, define
\begin{equation}\label{eq:kappadef}
\bmu^i_j(\alpha) = \sum_{\sigma \colon \alpha^{\sigma}_{i}|_{Y'_j} = \alpha} \blambda_i(\sigma).
\end{equation}

\medskip
\noindent
{\bf Claim:}
Definition (\ref{eq:kappadef}) is independent of the choice of $X_i \in X_{(\leq 2k)}(Y'_j)$.
That is,
\begin{equation}\label{eq:toprove}
\bmu^{r}_j = \bmu^{i}_j \qquad \forall r, i \in [q'] \text{ such that } X_r, X_i \in X_{(\leq 2k)}(Y'_j).
\end{equation}

\begin{proof}[Proof of Claim]

First, we prove (\ref{eq:toprove}) for $X_r \subseteq X_i$ with $X_r \in X_{(\leq k)}(Y'_j)$ and $X_i \in X_{(\leq 2k)}(Y'_j)$.
We have
\begin{equation*}
\begin{aligned}
\bmu^{r}_j(\alpha)
& \makebox[2em]{$\stackrel{\text{(\ref{eq:kappadef})}}{=}$}
\sum_{\tau \colon \alpha^{\tau}_{r}|_{Y'_j} = \alpha} \blambda_{r}(\tau) \\
& \makebox[2em]{$\stackrel{\text{(\ref{las:vec-marg})}}{=}$}
\sum_{\tau \colon \alpha^{\tau}_{r}|_{Y'_j} = \alpha} \enspace \sum_{\sigma \colon \Crestrict{\sigma}{X_r}=\tau} \blambda_{i}(\sigma)\\
& \makebox[2em]{=}
\sum_{\sigma \colon \alpha^{\sigma_r}_{r}|_{Y'_j} = \alpha} \blambda_{i}(\sigma)\\
& \makebox[2em]{$\stackrel{\text{(c)}}{=}$} 
\sum_{\sigma \colon \alpha^{\sigma}_{i}|_{Y'_j} = \alpha} \blambda_{i}(\sigma) \\
& \makebox[2em]{$\stackrel{\text{(\ref{eq:kappadef})}}{=}$}
\bmu^{i}_j(\alpha),
\end{aligned}
\end{equation*}

Next, let $X_r \in X_{(\leq 2k)}(Y'_j)$ and $X_i \in X_{(\leq 2k)}(Y'_j)$ be arbitrary.
From~(\ref{eq:thindown}), it follows that $X_r$ contains a subset $X_{s} \in X_{(\leq k)}(Y'_j)$
and that $X_i$ contains a subset $X_{t} \in X_{(\leq k)}(Y'_j)$.
Since $\left| X_{s} \cup X_{t} \right| \leq 2k$,
there exists an index $u \in [q']$ such that $X_{u} = X_{s} \cup X_{t}$.
The claim (\ref{eq:toprove}) now follows by a repeated application of
the first case: $\bmu^{r}_j = \bmu^{s}_j = \bmu^{u}_j = \bmu^{t}_j = \bmu^{i}_j$.
\renewcommand{\qedsymbol}{$\blacksquare$}\end{proof}

By (\ref{eq:toprove}), we can pick an arbitrary $X_i \in X_{(\leq 2k)}(Y'_j)$ and
uniquely define $\bkappa_j = \bmu^i_j$. 

We now show that this definition of $\bkappa$ satisfies the equations (\ref{las:1})--(\ref{las:split}).
Similarly to the definition of $\blambda_0$, 
we let $\bkappa_0$ be a shorthand for $\bkappa_j(\emptyset)$, where
$j$ is the index for which $Y'_j = \emptyset$.

\begin{itemize}
\item
The equation~(\ref{las:1}) holds as $\bkappa_0 = \sum_{\sigma} \blambda_i(\sigma) = 1$
for an arbitrary $i$ by (\ref{las:vec-marg}).
\item 
The equations~(\ref{las:nonnegative}) holds by the linearity of the inner product.
\item 
The equations~(\ref{las:infeas}) hold trivially if $\phi'_j$ is a null constraint.
Otherwise, $j \in C_i$ for some $i \in [q]$.
This implies that $X_i \in X_{(\leq k)}(Y'_j)$, and by (\ref{eq:toprove}) we have $\bkappa_j
= \bmu_j^i$.
Then, $\alpha \in \bsupp(\bkappa_j)$ implies that there is a $\sigma \in \bsupp(\blambda_i)$
such that $\alpha^{\sigma}_{i}|_{Y'_j} = \alpha$.
By condition~(\ref{cond:itoj}) and equation~(\ref{las:infeas}) for $\blambda_i$,
the tuple $\alpha^{\sigma}_{i}(\tup{y}'_j) \in \feas(\phi'_j)$,
so $\bkappa_j$ satisfies (\ref{las:infeas}).
\item
We show that the equations~(\ref{las:var}) hold for $\bkappa$. Let $Y'_j=\{y\}$ be a
singleton and let $X_i\in X_{(\leq k)}(Y'_j)$.
We have
\begin{equation*}
\begin{aligned}
& \sum_{a'\in D'} ||\bkappa_j(a')||^2\\ 
& \makebox[2em]{$\stackrel{\text{(\ref{eq:kappadef})}}{=}$}
\sum_{a'\in D'} \langle \sum_{\sigma \colon \alpha^{\sigma}_{i}(y)=a'} \blambda_i(\sigma), \sum_{\sigma \colon \alpha^{\sigma}_{i}(y)=a'} \blambda_i(\sigma) \rangle \\
& \makebox[2em]{$\stackrel{\text{(\ref{las:cons})}}{=}$}
\sum_{a'\in D'} \sum_{\sigma \colon \alpha^{\sigma}_{i}(y)=a'} \langle  \blambda_i(\sigma), \blambda_i(\sigma) \rangle \\
& \makebox[2em]{=}
\sum_{\sigma} \langle  \blambda_i(\sigma), \blambda_i(\sigma) \rangle \\
& \makebox[2em]{=}
||\sum_{\sigma} \blambda_i(\sigma)||^2 \\ 
& \makebox[2em]{$\stackrel{\text{(\ref{las:vec-marg})}}{=}$} 
|| \blambda_0||^2 \\
& \makebox[2em]{$\stackrel{\text{(\ref{las:1})}}{=}$}
1.
\end{aligned}
\end{equation*}

\item 
The equations~(\ref{las:cons}) hold by linearity of the
inner product and by the equations~(\ref{las:cons}) for $\blambda$.
\item
Finally, we show that the equations~(\ref{las:split}) hold for $\bkappa$.
Let $r, s \in [p]$, and pick assignments $\alpha_r \colon Y'_r \to D'$,
$\alpha_s \colon Y'_s \to D'$.
From (\ref{eq:thindown}) it follows that there are 
$X_u \in X_{(\leq k)}(Y'_r)$ and 
$X_t \in X_{(\leq k)}(Y'_s)$.
Then, there is an index $i \in [q']$ such that $X_i = X_u \cup X_t$.
It follows that $X_i \in X_{(\leq 2k)}(Y'_r)$ and $X_i \in X_{(\leq 2k})(Y'_s)$.
Therefore,
\begin{equation}\label{circ}
\begin{aligned}
& \langle \bkappa_r(\alpha_r), \bkappa_s(\alpha_s) \rangle \\
& \makebox[2em]{$\stackrel{\text{(\ref{eq:kappadef})}}{=}$} 
\langle \sum_{\sigma \colon \alpha^{\sigma}_{i}|_{Y'_r}=\alpha_r} \blambda(\sigma), \sum_{\sigma' \colon \alpha^{\sigma'}_{i}|_{Y'_s}=\alpha_s} \blambda(\sigma') \rangle \\
& \makebox[2em]{=}
\sum_{\sigma \colon \alpha^{\sigma}_{i}|_{Y'_r}=\alpha_r} \sum_{\sigma' \colon \alpha^{\sigma'}_{i}|_{Y'_s}=\alpha_s} \langle \blambda(\sigma), \blambda(\sigma') \rangle\\
& \makebox[2em]{$\stackrel{\text{(\ref{las:cons})}}{=}$}
\sum_{\sigma \colon \alpha^{\sigma}_{i}|_{Y'_r \cup Y'_s}=\alpha_r \circ \alpha_s} \langle \blambda(\sigma), \blambda(\sigma) \rangle\\
\end{aligned}
\end{equation}
Now, let $r', s' \in [p]$ be such that $Y'_r\cup Y'_s = Y'_{r'}\cup Y'_{s'}$
and $\alpha_{r'} \colon Y'_{r'} \to D', \alpha_{s'} \colon Y'_{s'} \to D'$ be such that
$\alpha_r \circ \alpha_s = \alpha_{r'} \circ \alpha_{s'}$.
Then, the right-hand side of (\ref{circ}) is identical for
$\langle \bkappa_r(\alpha_r), \bkappa_s(\alpha_s) \rangle$ and
$\langle \bkappa_{r'}(\sigma_{r'}), \bkappa_{s'}(\sigma_{s'}) \rangle$.
\end{itemize}
We conclude that $\bkappa$ is a feasible solution to the $\Las(k')$-relaxation of $J$.

Let $i \in [q]$ and note that by (\ref{eq:toprove}), for every $j \in
C_i$, we have $\bkappa_j = \bmu_j^i$. 
Therefore,
\begin{equation}
\begin{aligned}
& \sum_{j \in C_i} \sum_{\alpha \in \feas(\phi'_{j})} ||\bkappa_{j}(\alpha)||^2 \phi'_j(\alpha(\tup{y}'_j)) \\
&=
\sum_{j \in C_i} \sum_{\alpha \in \feas(\phi'_{j})} \sum_{\sigma \colon \alpha^{\sigma}_{i}|_{Y'_j} = \alpha} ||\blambda_{i}(\sigma)||^2 \phi'_j(\alpha(\tup{y}'_j))\\
&=\sum_{\sigma \colon \alpha^{\sigma}_{i}|_{Y'_j} \in \feas(\phi'_j)} ||\blambda_{i}(\sigma)||^2
\sum_{j \in C_i} \phi'_j(\alpha^{\sigma}_{i}(\tup{y}'_j))\\
&\leq
\sum_{\sigma \in \bsupp(\blambda_{i})} ||\blambda_{i}(\sigma)||^2 \phi_i(\sigma),
\label{eq:objfn}
\end{aligned}
\end{equation}
where the inequality follows from assumption~(\ref{cond:itoj}).
Summing inequality~(\ref{eq:objfn}) over $i \in [q]$ shows that 
$\sdpval(J,\bkappa,k') \leq \sdpval(I,\blambda,2k)$ and the lemma follows.
\end{proof}

\section*{Acknowledgements}

We thank the anonymous reviewers of both the conference version~\cite{tz17:lics}
and (in particular) the journal version of this paper for their comments.

\newcommand{\noopsort}[1]{}\newcommand{\Zivny}{\noopsort{ZZ}\v{Z}ivn\'y}

\end{document}